\documentclass[11pt]{article}
\usepackage{booktabs} 
\usepackage[ruled]{algorithm2e} 

\SetAlFnt{\small}
\SetAlCapFnt{\small}
\SetAlCapNameFnt{\small}
\SetAlCapHSkip{0pt}
\IncMargin{-\parindent}

\usepackage[T1]{fontenc}
\usepackage{lmodern}

\newcounter{protocol}

\makeatletter
\newenvironment{protocol}[1][htbp]
  {%
    \let\c@algocf\c@protocol
    
    \begin{algorithm}[#1]
  }
  {%
    \end{algorithm}
  }
\makeatother

\usepackage{graphicx} 
\usepackage{comment}
\usepackage{array}

\usepackage{xcolor}

\usepackage{fullpage}
\usepackage{booktabs}
\usepackage{amsthm}

\usepackage{amsmath}
\usepackage{amsthm}
\usepackage{amssymb}
\usepackage{color,soul}

\usepackage{hyperref}
\usepackage{thmtools}
\usepackage{thm-restate}
\usepackage{multirow}
\usepackage{cleveref}
\usepackage{natbib}

\newtheorem{theorem}{Theorem}
\newtheorem{lemma}{Lemma}
\newtheorem{observation}{Observation}
\newtheorem{definition}{Definition}
\newtheorem{corollary}{Corollary}
\newtheorem{fact}{Fact}
\newtheorem{proposition}{Proposition}

\newtheorem{question}{Open Question}
\newtheorem{remark}{Remark}

\newcommand{\alg}{\textsc{ALG}}
\newcommand{\poly}{\textsc{POLY}}
\newcommand{\opt}{\textsc{OPT}}
\newcommand{\ac}{\textsc{AC}}

\newcommand{\mwlm}{\textsc{MWLM}}
\newcommand{\mwlbm}{\textsc{MWLBM}}
\newcommand{\wobm}{\textsc{WOBM}}
\newcommand{\supp}{\textsc{supp}}
\newcommand{\child}{\textsc{CHILD}}

\newcommand{\1}[1]{\mathbf{1}\left\{#1\right\}}
\newcommand{\inner}[1]{\left\langle #1 \right\rangle}
\newcommand{\ceil}[1]{\left\lceil #1 \right\rceil}

\newcommand{\E}{\mathbb{E}}
\newcommand{\R}{\mathbb{R}}
\newcommand{\calA}{\mathcal{A}}

\newcommand{\calO}{\mathcal{O}}

\newcommand{\calS}{\mathcal{S}}

\newcommand{\hf}[1]{{  }}
\newcommand{\adi}[1]{{  }}

\usepackage{xcolor}
\title{The Keychain Problem: On Minimizing the Opportunity Cost of Uncertainty}
\author{
    Ramiro N. Deo-Campo Vuong
    \thanks{
        Cornell University.
        Email: {\tt ramdcv@cs.cornell.edu}
    }
    \and Robert Kleinberg
    \thanks{
        Cornell University.
        Email: {\tt rdk@cs.cornell.edu}
    }
    \and Aditya Prasad
    \thanks{
        The University of Chicago.
        Email: {\tt adityaprasad@uchicago.edu}
    }
    \and Eric Xiao
    \thanks{
        National Taiwan University.
        Email: {\tt b10401011@ntu.edu.tw}. Part of the work was completed while this author was visiting the Computer Science Department at the University of Chicago.
    }
    \and Haifeng Xu
    \thanks{
        The University of Chicago.
        Email: {\tt haifengxu@uchicago.edu}
    }
}

\begin{document}

\maketitle

\begin{abstract}
In this paper, we introduce a family of sequential decision-making 
problems, collectively termed the \emph{Keychain Problem}, that involve exploring a set of actions to maximize expected payoff when only a subset of actions are available in each 
stage.
In an instance of the Keychain Problem, a locksmith faces a sequence of decisions,
each of which involves selecting one key from a keychain (a subset of keys)
to attempt to open a lock. Given a Bayesian prior on the effectiveness
of keys, the locksmith's goal is to 
minimize the \emph{opportunity cost}, which is the expected number of rounds in which the chain has a \emph{correct key} but our selected key is incorrect. 

We study the computation of the Bayes optimal solution for Keychain Problems. Employing polynomial-time reductions, we establish formal connections between natural variants of the Keychain Problem and  well-studied algorithmic economics problems on bipartite graphs. When the keychain order is known to the locksmith, we show that it reduces to Maximum Weight Bipartite Matching (MWBM). More general  is the situation when the keychain order is sampled from a prior distribution (possibly correlated with the correct key). Here  the Keychain Problem reduces to a novel generalization of MWBM which we coin the \emph{Maximum Weight Laminar Matching}, which then further reduces to combinatorial auctions under XOS valuation functions. Finally, we show that when the locksmith can choose the keychain order, the Keychain problem reduces from a classic NP-hard combinatorial problem, again, on bipartite graphs. Besides implying  algorithmic results and deepening our structural understanding about the Keychain Problem, our established reductions also find applications beyond --  for example, to the Philosopher Inequality for online bipartite  matching.


 
\end{abstract}



\newpage
\section{Introduction}\label{sec:intro}
Designing experimentation policies to 
optimally explore a space of alternatives 
has been a fruitful theme for algorithms research,
encompassing influential models such as multi-armed bandits
\cite{gittins2011multi,lattimore2020bandit,slivkins2019bandit} and Pandora's problem~\cite{kleinberg2016descending,singla2018price,weitzman}. An aspect of optimal
exploration problems that has been under-examined in prior
work is \emph{limited availability}: the set of alternatives
available for exploration may vary over the course of the 
exploration process. A notable exception is the
\emph{sleeping bandits} problem \cite{blum2007external,kleinberg2010regret} where the set of available bandits at each round is prior-independent knowledge.
For the more general Bayesian settings that integrate  knowledge over available actions as a prior distribution, we are not aware of any prior work addressing the
design and analysis of optimal  
experimentation policies under limited availability.
This paper seeks to fill that gap.

We introduce and study a problem of sequential decision-making
under uncertainty that we call the \emph{Keychain Problem}. 
In this setting, there is a locksmith 
tasked with opening a lock.
At her disposal is a sequence of keychains.
In an online fashion, the locksmith irrevocably selects 
at most one key
from the next keychain in the sequence. She observes if the selected
key opens the lock, and receives a reward if it does.
The goal of the locksmith is to quickly identify the correct
keys so she can recover rewards on subsequent keychains.
We measure the performance of the locksmith against the
performance of a Bayes optimal policy.

There are essentially two reasons why optimal experimentation
in the face of limited availability is difficult. First, since
the set of actions available varies over time, one must seize
opportunities to reduce uncertainty about the identity of the
best action in light of the choices one will face in the
future (see a concrete example in
\Cref{subsec:intro_example}.) 
Second, since the best action may not be available at all times, one must also work to identify 
high-performing alternatives for the 
times when it is unavailable. The Keychain Problem 
is abstracted to distill the first of these challenges
while rendering the second one moot by equalizing the values
of all suboptimal actions.

Similar to Pandora's Problem \cite{weitzman}, 
the Keychain Problem admits a basic version whose 
optimal solution is  tractable (see \Cref{subsec:intro_warmup})
and a wealth of versions that are algorithmically rich.  
Our paper is  devoted to exploring structural understanding of natural variants of the Keychain Problem; we do so by employing efficient reductions to and from well-studied problems in algorithmic economics. Besides implying algorithmic properties of Keychain Problems,   
we hope these connections constitute an enticing invitation 
for researchers to further investigate this setting.  
We conclude by listing a rich set of interesting future directions.

\subsection{Example: Advisor Search}\label{subsec:intro_example}
Consider the problem of finding an advisor as a first-year graduate student.
Lily (locksmith) is a new PhD student at her dream university.
There are three rotational periods (keychains) in her first year.
For each rotation, she adaptively selects a single professor to collaborate with.
After working with a professor during a rotation, she learns if the professor is a good advising fit and may continue to work with the same professor for future rotations.
To remain in good standing with the university, Lily must identify an advisor and partake in successful collaborations during this rotational period.
There are three professors (keys) Lily is interested in working with: Alice, Bob, and Carol.
The professors have very different advising styles and interests, so \textit{only one} professor will make a good advisor for Lily.
Based on the professors' previous work and accounts from other students, Lily estimates that Alice, Bob, and Carol will be a good advisor fit with probabilities of $3/7$, $2/7$, and $2/7$, respectively.
Bob and Carol are available for every rotation.
In contrast, Alice is considering a short sabbatical period and will not be available to collaborate during the second rotation with probability $2/3$, independent of which professor will make a good advisor.
Notably, Lily will not learn of Alice's realized availability for the second rotation before choosing who to work with in the first rotation.
If Lily's goal is to maximize her expected number of rotations with the professor who would make a good advisor, whom should she choose to work with at each rotational period?
\begin{table}[ht]
    \centering
    \begin{tabular}{lcccc}
    \toprule
    \multirow{2.4}{*}{\textbf{Advisor}} & \multirow{2.4}{*}{\textbf{Fit Probability}} & \multicolumn{3}{c}{\textbf{Availability Probability}} \\
    \cmidrule(lr){3-5}
     & & \textbf{Rotation 1} & \textbf{Rotation 2} & \textbf{Rotation 3} \\
    \midrule
    Alice & $3/7$ & $1$ & $1/3$ & $1$ \\
    Bob   & $2/7$ & $1$ & $1$ & $1$ \\
    Carol  & $2/7$ & $1$ & $1$ & $1$ \\
    \bottomrule
    \end{tabular}
    \caption{Advisor fit probabilities and availability across rotations. Alice is unavailable for Rotation 2 with probability $2/3$, independently of advisor fit. \vspace{-5mm}}
\end{table}

Let us consider two polices.
The first greedy policy selects the professor that yields the largest number of expected successful total rotations, conditioned on past observations.
The second policy strikes a balance between greedy selection and exploring alternatives to reduce uncertainty in future rotational periods.
For the sake of brevity, we leave explicit calculations for this example in Appendix \ref{sec:app_intro_example}.

The Greedy policy first collaborates with Alice because she yields the greatest number of successful total rotations in expectation.
In expectation, choosing to collaborate with Alice first yields $1$ successful total rotation, whereas collaboration with Bob or Carol first only offers $6/7$ successful total rotations.
If collaboration with Alice is successful during the first rotation, then Lily should continue to collaborate with Alice on subsequent rotations.
Otherwise, she should choose either Bob or Carol on the second rotation, and the other one on the third rotation, if the second-rotation collaboration is unsuccessful.
Lily can expect to have $13/7$ successful rotations by greedily selecting collaborators.

It turns out the greedy policy is not optimal in Lily's situation.
Counterintuitively, Lily should collaborate with Bob (or Carol by symmetry) for her first rotation.
The interesting event is when Bob is not the right advisor for Lily, and Alice goes on sabbatical. By ruling out Bob as the correct advisor on the first rotation, Lily can rotate with Carol next to learn if she is the correct advisor at no opportunity cost (since Carol is the only available advisor on the second rotation who could be a good fit). 
\emph{This reduction in opportunity cost turns out to outweigh the benefit of greedily collaborating with Alice first.} Specifically, the following   policy  can be shown to be optimal for Lily: she always collaborates with Bob for the first rotation  and adapts to Alice's availability on the second rotation. If Alice takes a sabbatical period (with prob. $2/3$), Lily  collaborates Carol second and   Alice at last; Otherwise, Lily should choose to collaborate with Alice second and Carol at last. Naturally, during this process, she stops rotation whenever the good advising fit is  identified. Simple calculations show that  Lily can expect  $40/21$  successful rotations with this  policy, which exceeds the $13/7$ expected successful rotations yielded by the greedy policy.


\subsection{Overview of Results and Their Applications} 
Our primary contribution, both technically and conceptually, is to establish formal connections between Keychain Problems and several widely-studied optimization problems in algorithmic economics; we do so by employing polynomial-time reductions that also help reveal structural similarities and differences between Keychain Problems and other optimization problems. Moreover, we further demonstrate that these established reductions not only imply efficient polynomial-time algorithms for various versions of Keychain Problems, but also find applications, and help make progress on, other (seemingly unrelated) algorithmic problems.

Overall, our results  reveal  close connections between Keychain Problems and problems on bipartite graphs. We start from the   basic setting with known keychain order and show an efficient reduction from computing the optimal Bayes solution to finding the maximum weight bipartite matching (MWBP). The crux of our investigation is for the  more general Bayesian setting where the keychain order is unknown and the designer has a probabilistic belief about the order that is allowed to correlate with the (also unknown) correct key. We show that this Bayesian version of Keychain Problems is closely related to a natural generalization of MWBP, which we term the \emph{maximum weight laminar matching} (MWLM). In MWLM, each right-side node $j$ in the bipartite graph is allowed to have a \emph{type set} $T_j$, and the collection of typesets over all right-nodes forms a laminar set family.\footnote{ A  set family $T_1, \dots, T_n \subseteq \mathcal{T}$   over a ground type set   $\mathcal{T} = \{t_1, \dots, t_k\}$  is called a \emph{laminar family} if,  for all $i$ and $j$, we have  $T_i \cap T_j = T_i$, or $T_i \cap T_j = T_j$, or $T_i \cap T_j = \emptyset$.} An edge set is called a \emph{laminar matching} if it is a standard bipartite matching for any subgraph induced by any type $t$ by deleting all right nodes $j$ such that  $ t \not \in T_j$.\footnote{When $T_j = \{ t \}$ is the same singleton set for every $j$, this degenerates to standard bipartite matching.}  

For the above Bayesian version of Keychain Problems, we show two reductions. First, computing the optimal Bayes optimal solution reduces to the MWLM problem described above. Second, MWLM furter reduces to the classic combinatorial auction with XOS valuation problem (see, e.g., \cite{dobzinski2005auctions,dobzinski2010auctions}). Both reductions are 
in polynomial time and, notably, are approximation-preserving. As a corollary,  these reductions immediately imply a $(1-1/e)$-approximate algorithm for the corresponding Keychain problem since there is a $(1-1/e)$-approximation for combinatorial auction with XOS valuation \cite{dobzinski2010auctions}. More interestingly, these connections  also find applications in  classical combinatorial problems. We demonstrate this by examining the widely studied \emph{online} weighted bipartite matching problem. The well-known Philosopher Inequality  captures the approximation ratio of a polynomial-time online algorithm against the Bayes optimal algorithm \cite{papadimitriou2021online}.  Recent results by \cite{braverman2025new} show a $0.678$ approximation for the Philosopher Inequality, assuming edges are drawn independently. A straightforward application of our reductions implies a $(1-1/e)$-approximation (slightly worse than $0.678$) for  the more general settings that allows $b$-matching and allows edge correlations in the randomness (though requiring the support of the randomness to have polynomial size). 
 We conclude this section by proving the NP-hardness of obtaining a better-than-$(\frac{4063}{4064}+\epsilon)$ approximation for the Keychain problem with a probabilistic belief of keychain orders, ruling out any PTAS. 
 
Our last   section  turns to another natural variant of Keychain Problems where the keychain order can be optimized  by the desiger. Our result reveals connection to another bipartite graph problem which we call \emph{downward-facing bipartite permutation} --  given any  bipartite graph with directed edges pointing from left nodes to right ones, decide whether there is a way to independently permute the left and right side of nodes to make all edges weakly downward-facing. This is a known NP-hard problem which we show reduces to our Keychain Order Selection setting, hence implying its NP-hardness. Nevertheless, we present a clean $(1/2)$-approximate algorithm. 

Finally, we conclude with many open directions, together with some interesting initial findings that could hint at plausible future paths. 

\subsection{Keychain Problems in Application Domains}
 

Keychain Problems model the wide-ranging task of opportunity-cost minimization in search environments where each opportunity provides a limited option set. While Section \ref{subsec:intro_example} provides a natural (though perhaps toy) application to advisory search,  we now exemplify two more realistic application domains.

\paragraph{Clinical trials.} The problem of designing clinical trials to identify the best treatment for a disease (i.e., the ``lock'' to be opened) is a key application of online exploration (see e.g. \cite{guha2007approximation,slivkins2019bandit}). There are typically multiple candidate treatments (i.e., keys), each with some probability of being the best, which is captured as the designer's prior belief. Participating patients (i.e., keychains) arrive in order, or their order could be decided by the treatment designer. These correspond to optimizable keychain orders or pre-determined orders. In our models, the patient order can be correlated with the (random) best treatment option. 
Notably, the applicable treatments to each patient are usually a subset of candidate treatments, corresponding to each keychain containing a subset of keys, due to patients' reluctance to some candidate treatments at a trial stage \cite{taylor1984physicians} or potential treatment incompatibility issues \cite{ross1999barriers}.  In such applications, the designer must decide an applicable treatment for each patient, and possibly the order of patients for trials, in order to find out the right or  best treatment as soon as possible. 

\paragraph{Cargo transportation in hostile environments via camouflaging.} 
A completely different application domain  with the same underlying problem structure  is  critical cargo   transportation. 
The task is to transport critical cargo across contested waters while avoiding detection by adversaries. Camouflage technologies are widely used in such applications \cite{umich_camouflage,boehmer2024escape}. Specifically, 
each shipment vessel (keychain) is equipped with a set of compatible camouflage technologies (keys) that it may deploy. 
The transportation organization only has a prior belief about the capabilities of these camouflage technologies against the unknown adversary. To avoid detection by adversarial parties, the shipment planner chooses one available camouflage technology (key) to activate for each vessel, subject to compatibility constraints between camouflage technologies and vessels.
In this setting, avoiding detection corresponds to opening the lock and the vessel order is typically optimizable. 

\subsection{Related Work} \label{subsec:related_work}

The Keychain Problem relates to several areas of literature on sequential decision-making under uncertainty, stochastic search, and online matching.
Although these settings share our overall motivation of optimal search, below we illustrate how these problems and their algorithmic challenges differ from ours. 

Our theme of exploring a space of actions subject to constraints
on the available actions   connects  to the 
\emph{Sleeping Bandits} problem~\cite{blum2007external,kleinberg2010regret}.
While sharing similar motivation of accounting for limited availability, the two problem formulations
are  very different: in the Keychain Problem one has a prior
distribution on the available actions at each round and the identity of the correct key. 
In the sleeping bandits problem there is no prior knowledge about  available bandits. 
Accordingly, the two problems are
solved using very different techniques, as will   be illustrated by our results.  

More generally, in stochastic search problems, it is rare for the
Bayes optimal procedure is to be computationally efficient. 
One of the earliest examples of an efficiently solvable stochastic
search problem is the celebrated Pandora's Problem, formulated by ~\cite{weitzman} who also
showed that a remarkably simple ``index policy'' is Bayes optimal. 
Subsequent work in theoretical computer science obtained approximately  optimal policies for combinatorial generalizations 
of Pandora's Problem in which more than one prize may
be claimed~\cite{singla2018price}, decision-theoretic
generalizations in which there is more than one way
to open a box~\cite{beyhaghi2023pandora,beyhaghi2019pandora,bowers2025prophetinequalitiesbanditscabinets,chawla2024combinatorialselectioncostlyinformation,fu2023pandora,guha2008information,scully2025localhedgingapproximatelysolves},
and strategic generalizations in which boxes are under
the control of selfish agents~\cite{kleinberg2018delegated,kleinberg2016descending}.
Our Keychain Order Selection problem bears a resemblance to
Pandora's Problem in that the searcher has freedom, in both
models, to decide on the order in which to explore 
alternatives. However,    in  Keychain Order Selection this freedom is not absolute; it is limited by the structure of the keychains, which inhibits the application of 
techniques based on Weitzman indices and their generalizations.

Our techniques find applications to the  weighted online bipartite matching under the Philosopher Inequality benchmark, as studied by \cite{braverman2022max, braverman2025new, papadimitriou2021online}. 
Algorithms in this setting compete against a Bayes optimal policy computed by a philosopher with unlimited computational resources.
These works consider the case in which the incident edges on each online node are drawn independently from a known distribution.
In contrast, our techniques allow arbitrary stochastic models with few realizable edge sets under more general one-sided $b$-matching constraints.
\cite{papadimitriou2021online} shows that online bipartite matching with the Philosopher Inequality benchmark is $\textsc{PSPACE}$-hard to approximate beyond a constant factor, and \cite{braverman2025new} provides an algorithm with a current best $0.678$-approximation ratio via pivot sampling. 
Other works by \cite{naor2025online, saberi2021greedy} present rounding algorithms with applications to online matching in the  Philosopher Inequality paradigm.

On a deeper technical level, our results turn out to relate to various lines of work in combinatorial auctions, ordering constraint satisfaction problems, and Quadratic Assignment problems. We defer additional discussion on these topics to Appendix \ref{sec:app_related_works}.

\subsection{Preliminaries on Notations}
\paragraph{Notation Conventions}
We use $[n] = \{1,\dots, n\}$ to denote the set of $n$ elements.
The power set of set $S$ is denoted by $2^{S}$.
As a shorthand for an additive set function with weights $w_1,\dots,w_n$, we use $w(S) = \sum_{i\in S} w_i$.
We denote the collection of distributions whose sample space is the finite set $S$ by $\Delta(S)$.
For a vector $v$, the $i$th element is $v_i$.
Additionally, $v_{i:j}$ for $i\le j$ denotes the vector $(v_i,v_{i+1},\dots,v_j)$.
The element in the $i$th row and $j$th column of matrix $M$ is $M_{i,j}$.
To denote the elements in the first $j$ columns, we use $M_{:,1:j}$. To represent a ``NULL'' option, we use $\perp$.
The indicator function is $\1{\cdot}$; the output is $1$ when the input predicate is true and $0$ otherwise.
Finally, $\inner{\cdot}$ denotes a tuple often used to describe an instance of a problem.

\textit{Approximation Guarantees. } 
In many of our settings, it is intractable to compute a Bayes optimal solution.
In these cases, we consider multiplicative and mixed multiplicative-additive approximations.
An algorithm is said to be an $\alpha$-approximation if it outputs a policy with reward $\alg$ that competes with the reward of a Bayes optimal solution $\opt$:
 $\alg \ge \alpha \cdot \opt$. 
More generally, an algorithm is a $(\alpha,\epsilon)$-approximation if it outputs policies satisfying: $\alg \ge \alpha \cdot \opt - \epsilon$.


\section{The Keychain Problem:  Basic Settings  and Preliminary Results} \label{sec:warm-up}
We start by formally introducing the \emph{Keychain Problem}.
There is a set of \emph{keys} $[n]$, and a single lock.
A \emph{locksmith} -- our algorithm designer -- is presented with a sequence of \emph{keychains} $C_1,\dots, C_m$, each being a subset of $[n]$. 
There could be one or multiple \emph{correct keys} (i.e., lock-opening keys) from $[n]$ to open the lock. The locksmith does not know which keys are correct, and only possesses some prior knowledge that is modeled as a Bayesian prior distribution $p$ on the probability of keys' correctness (possibly jointly with other setting scenarios). The locksmith hopes to find a correct key from every keychain $C_t$ (if $C_t$ has one) to open the lock.  
The tuple $\langle n, \{ C_t \}_{t=1}^m, p \rangle$ hence consists of the input to a Keychain Problem.

As it turns out, the tractability of the Keychain Problem depends on the number of correct keys as well as whether the order of the keychains is fully known,  drawn from a random prior, or  an optimizable decision. We start from the basic setting with a single correct key and known keychain order     in this section, and will   gradually generalize this setup in the following sections.
In this basic case, the locksmith's prior knowledge encodes a distribution $p = (p_1, \cdots, p_n) \in \Delta([n])$ in which $p_i$ is the probability that key $i$ is the correct key; let $k^*$ denote the correct key, drawn from the distribution $p$.
Given her prior knowledge and order of keychains $C_1, C_2, \cdots, C_m$, the locksmith must design a (possibly history-dependent) \emph{policy} $\pi$ that irrevocably selects a key $k_t \in C_t \cup\{\perp\}$  at time step $t=1,\dots, m$.  She then observes whether or not  $k_t =k^* $  and receives a reward   $\1{k^* = k_t}$, which will affect how the policy $\pi$ selects the next keys. For mathematical  convenience,  $k_t = \perp$ means the locksmith  discards the keychain without trying any key.  
In   applications, a keychain captures an opportunity with multiple options (i.e., the keys), and will be discarded (i.e., the opportunity is gone) once one of its keys is tested. 
The locksmith looks to maximize   expected reward:  
\begin{align*}
    r(\pi) = \E_{p, \pi} \left[\sum_{t=1}^{m} \1{ k_t = k^*}\right]
\end{align*}
where the expectation is taken with respect to the randomness of the environment described by prior distribution $p$ and any randomness of the locksmith's policy $\pi$.

\begin{remark} More generally, each chain $C_t$ may have a non-negative weight $w_t$ capturing its importance, hence the objective $r(\pi) $ is generalized to $ \E_{p, \pi} \left[\sum_{t=1}^{m}  w_t \1{ k_t = k^*}\right]$. All our results in this paper can be straightforwardly extended to this weighted version. To avoid unnecessary notation, we focus our analysis on the unweighted objective above. \end{remark} 

Next we make some helpful observations that will simplify our search for the optimal policy $\pi$. Notably, when there is only a single correct key, there always exists a Bayes optimal policy that  always tries to play the correct key once it is identified. We call such policies \textit{exploitative}, and provide a formal definition below that is generalizable to multiple correct keys.
 

 \begin{restatable}[Exploitative Policies]{definition}{exploitation}
Let $K^*_\tau$ denote the (possibly empty) set of identified correct keys \emph{before} time $\tau$. We say a policy is exploitative if it always plays some $k_{\tau} \in K^*_\tau \cap C_{\tau}$ for any $\tau$ whenever $K^*_\tau \cap C_{\tau}$ is non-empty.
\end{restatable} 


\begin{observation}
\label{ob1:exploitative}
For the Keychain Problem with one correct key, there is a Bayes optimal policy that is deterministic and exploitative. 
\end{observation}

As a consequence of Observation \ref{ob1:exploitative} we can restrict our analysis to deterministic and exploitative policy for the remainder of this paper  when the setting has one correct key. Observation \ref{ob1:exploitative} is intuitive.
What is somewhat less intuitive, however, is that this same insight ceases to be true when there are multiple correct keys -- we present such a counterexample in Appendix \ref{appendix:multi_exploit_suboptimal}, which also hints at the non-triviality of settings with many correct keys even with a known keychain order. 

The proof of Observation \ref{ob1:exploitative} (deferred to Appendix \ref{sec:app_obs1}) is based on a reformulation of the keychain problem with one correct key as a Markov Decision Process (MDP), albeit with exponentially many states. The observation follows since any MDP admits optimal policies that are deterministic  \cite{puterman2014markov}. 

\subsection{The Basic Keychain Problem as Weighted Bipartite Matching}\label{subsec:intro_warmup}




A key conceptual contribution of this work is to establish  intrinsic connections between the Keychain Problem and other classic algorithmic problems. 
We start from connecting this basic setting of the keychain problem to   bipartite matching. 

Observation \ref{ob1:exploitative} suggests that it is without loss of generality to search for exploitative and deterministic policies. Because these policies act predictably after identifying the correct key, we model a policy by the keys it tests when it has not yet identified the correct key.
We encode deterministic and exploitative policies as a function $\pi: [m]\to [n] \cup \{ \perp \}$ from keychains to keys that are tested for the first time.
Since each key can be tested for the first time at most once, $\pi$ must be injective on $[n]$. 
This hints at a bipartite matching structure that we formalize below. 

    
    


\begin{proposition}\label{prop:intro_fixedalgo}
Computing the Bayes optimal policy for the Keychain Problem with known keychain order and one correct key reduces in $O(nm^2)$ time to   maximum weighted bipartite matching. 
\end{proposition}

\begin{proof}[Proof Sketch.] Let $  \langle n, \{ C_t \}_{t=1}^m, p \rangle$ denote the keychain problem instance, and $\pi$ be an arbitrary deterministic and exploitative policy. Construct a bipartite graph with left-side nodes as all keys and right-side nodes as all keychains. An edge $e=(k, t)$ exists if and only if key $k \in C_t$. Since there is only one correct key, the policy $\pi$ can be fully specified by the first keychain $C_t$ that the key $k$ is tried, corresponding to picking the edge $(k, t)$ in the constructed bipartite graph. Since any key only has one chance to be tried for the first time and any chain can only be explored once, all these edges induced by the policy $\pi$ as above forms a bipartite matching.  

If policy $\pi$ selects key $k$ on keychain $C_t$, its future (expected) reward is: $ 
    r_{k,t} =  \Pr_{k^*\sim p}[k^* = k]   \cdot  \left(\sum_{\tau=t}^{m} \1{k\in C_{\tau}}\right)$. The crux of the proof is to show that      
   the  expected  reward of any policy $\pi$ is precisely the total weight  of its corresponding matching, with $r_{k, t}$ as the weight of edge $(k,t)$. Intuitively, this is due to  the law of total expectation that adds up the reward $r_{k, t}$ from each edge $(k,t)$ chosen by the policy choice. We defer formal arguments to Appendix \ref{append_sec_reduction_mwbp}.    
\end{proof}

Proposition \ref{prop:intro_fixedalgo} implies an efficient algorithm for computing the Bayes optimal solution. We remark that such tractability does not extend to situations with many correct keys. As mentioned above, if many keys are correct, exploitative policies are no longer necessarily optimal. In  Appendix \ref{app_subsec:many-key}, we further show the APX-hardness of finding the Bayes optimal policy with a succinctly described prior over   correct keys.


\section{Keychain Problems with Probabilistic Chain Orders}\label{sec:scenarios}

This section investigates a natural generalization of the basic setting in Section \ref{sec:warm-up}. That is, the keychain order is not fully known; instead the locksmith only has a probabilistic prior about the order.

The locksmith now has probabilistic prior on both the correct keys  and  the keychain  orders. To capture their potential correlations,  we conveniently refer to the joint event of (correct key, keychain order) as a \emph{scenario} of the problem.\footnote{The term ``scenario''  is   commonly used term for similar purposes in  stochastic search problems; see e.g. \cite{heitsch2009scenario}.}     Let $\calS$ denote the support of all possible scenarios. With slight abuse of notation, the locksmith's prior   $p\in\Delta(\calS)$ now is a distribution over the scenarios. 
Each scenario encodes a realization of the environment. Formally, scenario $s\in\calS$ has an associated number of rounds $m(s)$, a sequence of keychains $C_1(s),\dots,C_{m(s)}(s) \in 2^{[n]}\setminus \{\emptyset\}$, and a correct key $k^*(s)\in[n]$. Let $m=\max_{s\in\calS} m(s)$ denote the maximum number of rounds in any scenario.

Given the instance input $\langle n, \calS, p \rangle$, the locksmith participates in a sequential decision-making process.
This process begins with the environment privately drawing a scenario $s\sim p$ from the prior.
On each round $t$ of $m(s)$ total rounds,  the keychain $C_t(s)$ is presented to the locksmith.
The locksmith then irrevocably selects a key $k_t\in C_t(s)\cup\{\perp\}$ from the keychain to test or opts to test nothing. 
The environment then reveals to the locksmith if $k_t = k^*(s)$.
Finally, the locksmith obtains a reward of $\1{k_t = k^*(s)}$.
We summarize this process in Protocol \ref{prot:scenarios_setting}.

\begin{protocol}[ht]
    \KwIn{
        Keys $[n]$, scenarios $\calS$, and prior $p\in\Delta(\calS)$.
    }
    The environment privately draws a scenario $s\sim p$\;
    \For{each round $t=1,\dots,m(s)$}{
        The learner observes keychain $C_{t}(s)$\;
        The learner irrevocably selects a key $k_t \in C_{t}(s) \cup \{\perp\}$\;
        The learner receives and observes reward $\1{k_t = k^*(s)}$\;
    }
    
    \caption{\textsc{The Interaction Protocol with Probabilistic Scenarios} ($[n]$, $\calS$, $p$)}
    \label{prot:scenarios_setting}
\end{protocol}


The locksmith looks for a policy $\pi$ that maximizes her expected reward, which is
 $   r(\pi) = \E_{p,\pi}\left[ \sum_{t=1}^{m(s)} \1{k_t = k^*} \right] $ 
where the expectation is with respect to the randomness of the environment $p$ and the possible randomness of the locksmith's policy $\pi$.
We consider algorithms that run in time polynomial in the number of keys $n$, the maximum keychain sequence length $m$, and the number of scenarios $|\calS|$.

\subsection{Two Related Problems and the Main Result}  
We now describe two  combinatorial problems that are  of natural interest in algorithmic economics. 
The first  is a familiar one.  
\begin{definition}[Combinatorial Auctions with XOS Valuations  (see, e.g., \cite{dobzinski2005auctions, dobzinski2010auctions})]\label{def:scenarios_comboauctions}
    The Combinatorial Auctions problem is defined on a set of buyers $[n]$ and a set of items $[m]$.
    Each buyer $i\in [n]$ has an associated normalized, monotone, and XOS valuation $v_i:[m]\to\R_{\ge 0}$;
    the auctioneer is assumed to have access to a value, demand, and supporting price oracle for each   valuation function $v_i$. 
    The goal is to compute an allocation (partition on items) $S_1,\dots, S_n$ that maximizes $v(S_1,\dots, S_n) = \sum_{i\in[n]} v_i(S_i)$.
\end{definition}


The second problem is a novel generalization of the classical  \emph{bipartite} matching, which we term the \emph{laminar matching} problem.  
Laminar matching augments classic bipartite matching on a bipartite graph $G = (L\cup R, E)$  by allowing each right node $j \in R$ to have a type set $T_j$.  
For any type $t$, let $R(t) = \{j\in R \mid t\in T_j\}$ denote the set of all right nodes including type $t$, and let $G(t) =   (L\cup R(t), E) $ denote the \emph{sub-graph} of $G$ induced by excluding all right-side nodes outside $R(t)$ as well as their corresponding edges.   Laminar matchings are then defined as follows. 
 
\begin{definition}[Laminar Matching]\label{def:scenarios_laminarmatch}
 Consider a bipartite graph $G = (L\cup R, E)$ and a collection of type sets $\{T_j\}_{j\in R}$ on right nodes that form a laminar family.\footnote{   A   set family $T_1, \dots, T_n \subseteq \mathcal{T}$ over the ground set $\mathcal{T}= \{t_1, \dots, t_k\}$ is called a \emph{laminar family} if, for all $i$ and $j$, we have $T_i \cap T_j = T_i$, or $T_i \cap T_j = T_j$, or $T_i \cap T_j = \emptyset$.}  An  edge set $M\subseteq E$ is called a laminar matching if $M$ restricted to $G(t)$ is a standard bipartite matching in the induced bipartite subgraph  $G(t)$ for every type $t$.  
\end{definition}



 It is easy to see that laminar matching  degenerates to the standard bipartite matching when $T_j  = \{ t \}$ is the same singleton set for every $j$.  

\begin{definition}[Maximum Weight Laminar Matching]\label{def:scenarios_mwlm}
 In the Maximum Weight Laminar Matching (MWLM) problem, one is given a bipartite graph $G=(L\cup R, E)$ with edge weights $w$ and a collection of type sets $\{T_j\}_{j\in R}$ on right nodes that form a laminar family. The goal is find a laminar matching $M$ with maximum weight $w(M) = \sum_{(i,j)\in M} w_{i,j}$.
\end{definition}

Our main result of this section is the establishment of the following formal connection between the Keychain problem and the two related problems described above. 
 
\begin{theorem}[Reductions from Keychain to MWLM, then to  Auctions. ]\label{thm:probalisitc-main} {\color{white}a} \\ \vspace{-4mm} 
 
\begin{itemize}
    \item    There is a  $O(\poly(n,m,|\calS|))$-time approximation-preserving reduction from computing the Bayes optimal policy for the Keychain Problem with probabilistic scenarios to the maximum weight laminar matching (MWLM) problem.  
    \item There is a  $O(\poly(n,m,|\calS|))$-time approximation-preserving reduction from   the maximum weight laminar matching (MWLM) problem to combinatorial auctions with XOS valuations.   
\end{itemize}
\end{theorem}
The "approximation-preserving" property means any $(\alpha, \epsilon)$-approximation algorithm for the latter can be converted in  $O(\poly(n,m,|\calS|))$-time to an $(\alpha, \epsilon)$-approximation algorithm for the former.  Recall that when the valuations of the buyers are XOS, normalized, and monotone, there are known $(1-1/e)$-approximation algorithms for the Combinatorial Auctions problem \cite{dobzinski2005auctions, dobzinski2010auctions}. As corollaries of Theorem \ref{thm:probalisitc-main}, we obtain a $(1-1/e)$-approximation for the corresponding keychain problem as well as the MWLM problem.

\subsection{Proof of Theorem \ref{thm:probalisitc-main}. }

\subsection*{Step 1:   Representing a Policy using Information Sets} 

Like the known keychain order problem with one correct key \ref{subsec:intro_warmup}, the Probabilistic Scenarios setting admits a deterministic and exploitative optimal policy.
However, the added complexity of the Probabilistic Scenarios setting requires more intricate policies beyond the simple injective mappings from rounds to keys that we previously employed.
Instead, we consider mappings between prefixes of realizable keychain sequences, called \emph{information sets}, to a previously untested key.

Before presenting our representation of a policy, we introduce some necessary notation for information sets.
An information set $o$ is the prefix of some realizable keychain sequence: $o = (C_1(s),\dots,C_t(s))$ for some scenario $s\in\calS$.
We use $\calO = \{(C_1(s),\dots,C_{t}(s))\}_{s\in\calS, t\le m(s)}$ to denote the collection of all information sets.
A scenario $s$ is said to be consistent with information set $o$ when information set $o$ is observed in scenario $s$ (e.g. $o = (C_1(s),\dots, C_t(s))$ for some round $t\le m(s)$).
The set $C(o)\subseteq\calS$ denotes the collection of all scenarios consistent with information set $o$.
We use $P(s) = \{(C_1(s),\dots,C_t(s))\}_{t\le m(s)}\subseteq\calO$ to denote the information sets for which scenario $s$ is consistent.
The information sets in $P(s)$ correspond with a chain (or root-to-leaf path) in the forest of trees representing information sets ordered by the prefix relation.
Finally, the probability of observing information set $o$ during the decision-making process is $p_o = \sum_{s\in C(o)} p_s$.

We represent policies as a mapping from information sets to keys $\pi: \calO \to [n]$.
We interpret $\pi(o) = k$ to mean that key $k$ is tested for the first time upon observing information set $o$, having not found a correct key before arriving at $o$.
Notably, since we consider exploitative policies, we require that admissible policies test each key at most once amongst all the information sets associated with a scenario.
Formally, we require that $\pi$ restricted to the information sets $P(s)$ must be injective for any scenario $s\in\calS$.
We do not explicitly require that such policies only select keys on a given chain because we ascribe $0$ reward to off-chain selections.
Protocol \ref{prot:scenarios_implementation} below describes how to implement an admissible policy $\pi$.

\begin{protocol}[ht]
    \KwIn{
        An admissible deterministic and exploitative policy $\pi:\calO \to [n]\cup\{\perp\}$.
    }
    Initialize $k^* = \perp$\;
    \For{each round $t=1,\dots,m$}{
        Observe chain $C_t$ and defined information set $o_t\gets (C_1,\dots, C_t)$\;
        \eIf{$k^* = \perp$}{
            Play $k_t = \pi(o_t)$ when $\pi(o_t) \in C_t$, otherwise play $k_t = \perp$\;
            If $k_t$ is correct, update $k^* = k_t$;
        }{
            Play $k_t = k^*$ when $k^*\in C_t$, otherwise play $k_t=\perp$\;
        }
    } 
    \caption{\textsc{Exploitative and Deterministic Policy Implementation}($\pi$)}
    \label{prot:scenarios_implementation}
\end{protocol}


Finally, we compute the expected reward of an admissible policy $\pi$.
Consider if an exploitative policy selects key $k$ upon observing information set $o = (C_1(s), \dots ,C_t(s))$ when scenario $s$ is realized.
The reward our policy would recover is
\begin{align*}
    r_{k,o,s}
&= \underbrace{\1{k\in C_t \text{ and }k = k^*(s)}}_{k\text{ is correct and selectable}} \underbrace{\left(\sum_{t'=t}^{m(s)} \1{k \in C_{t'}(s)}\right)}_{\text{future reward}}
\end{align*}
Notably, $r_{k,o,s} = 0$ when $o\notin P(s)$ or $k=\perp$.
Since we do not know the realized scenario, our algorithm cannot directly consider the above quantity.
Instead, we focus on the expected reward of selecting key $k$ at information set $o$, conditioned on scenarios consistent with our observations:
\begin{align*}
    r_{k,o}
    = \E\left[ r_{k,o,s} \middle| s\in C(o)\right]
    = \sum_{s\in C(o)} \left(\frac{p_s}{p_o}\right) r_{k,o,s}
\end{align*}
With these formalizations, we compute the expected reward of an admissible policy $\pi$:
\begin{align*}
    r(\pi) &= \E_{s\sim p} \left[ \sum_{k, o\in P(s)} r_{k,o,s}\1{\pi(o) = k} \right]
    = \sum_{s} \left( p_s \sum_{k, o\in P(s)} r_{k,o,s}\1{\pi(o) = k} \right) \\
    &= \sum_{k, o} \left(\left(\sum_{s\in C(o)} p_s r_{k,o,s}\right) \1{\pi(o) = k} \right) 
    = \sum_{k,o} p_o r_{k,o} \1{\pi(o) = k}
\end{align*}

\subsection*{Step 2:  From Keychain Problems to Maximum Weight Laminar Matching }
Recall from Definition \ref{def:scenarios_laminarmatch}, Maximum Weight Laminar Matching (MWLM) is defined on an edge-weighted bipartite graph and a set of types associated with each right node. 
The goal is to find a maximum weight laminar matching: an edge set in which each right node has at most one neighbor and each left node has at most one neighbor of each type.

In our reduction, keys correspond to left nodes, information sets correspond to right nodes, and types correspond to scenarios.
The types associated with a right node (information set) are exactly the set of scenarios consistent with the information set.
We define a bijection between policies and laminar matchings, in which an edge from a left node (key) $k$ to a right node (information set) $o$ encodes that the policy selects key $k$ on information set $o$ (e.g. $\pi(o) = k$).
By carefully selecting edge weights, the weight of any laminar matching is the expected reward of its corresponding policy.

\begin{lemma}[Reduction to Maximum Weight Laminar Matching]\label{lem:scenarios_redmwlm}
    For any $\alpha \in (0,1]$ and $\epsilon \ge 0$, if there is a polynomial time $(\alpha, \epsilon)$-approximation algorithm for the Maximum Weight Laminar Matching problem then there is a polynomial time $(\alpha, \epsilon)$-approximation algorithm for the Probabilistic Scenarios problem.
\end{lemma}

 \begin{proof} 
    Given an instance of the Probabilistic Scenarios problem $\inner{n, \calS, p}$, we construct an instance of MWLM $\inner{G = (L\cup R, E), w, T}$.
    The set of left nodes in the bipartite graph corresponds with keys $L = [n]$.
    Right nodes correspond with information sets $R = \calO$.
    The bipartite graph is complete (e.g. $E = L\times R$).
    The weight of edge $(k,o) \in L\times R$ is the expected reward of picking key $k$ upon observing information set $o$.
    Specifically, we assign weight $w_{k,o} = p_o r_{k,o}$.

    We now describe a bijection between admissible policies $\pi$ and laminar matchings, in which the expected reward of a policy is equivalent to the weight of the corresponding laminar matching.
    Admissible policy $\pi$ maps to the matching $M = \{(k,o)\in E \mid \pi(o) = k \}$.
    This mapping is bijective.
    Moreover, the expected reward of admissible policy $\pi$ and its corresponding laminar matching $M$ satisfy:
    \begin{align*}
        r(\pi)
        = \sum_{k,o} p_o r_{k,o} \1{\pi(o) = k}
        = \sum_{(k,o)\in M} w_{k,o}
        = w(M)
    \end{align*}
    The reduction takes polynomial time.
    Thus, a polynomial time $(\alpha,\epsilon)$-approximation algorithm for the Probabilistic Scenarios problem is implied by a polynomial time $(\alpha, \epsilon)$-approximation algorithm for Maximum Weight Laminar matching.
 \end{proof}

\subsection*{Step 3: From Laminar Matching to  XOS Combinatorial Auctions}\label{subsubsec:scenarios_combo}
Our third step develops a polynomial-time and approximation-preserving reduction from MWLM to Combinatorial Auctions with XOS valuations, as formally defined in Definition \ref{def:scenarios_comboauctions}. 

At a high level of our reduction, each left node is a buyer and each right node is an item.
An assignment of bundles corresponds with a laminar matching.
The constraint that each item is allocated once in an assignment guarantees that each right node has at most one neighbor in the laminar matching.
We use the valuations of buyers to encode the constraint that each left node is matched to at most one neighbor of each type in the laminar matching.
Specifically, the valuation of a buyer (left node) for a bundle of goods is the maximum weight subset of items (right nodes) in the bundle that do not share a type.
We call such valuations \textit{antichain} functions.
Finally, we specify a surjective map from allocations to laminar matchings in which the social welfare admitted by the allocation is the weight of the corresponding laminar matching.

\textbf{Preliminaries: }Our reduction necessitates background about set functions.
In the main body, we assume familiarity with previously studied XOS functions, monotone functions, normalized functions, value oracles, demand oracles, and supporting price oracles. The formal definitions of  these objects are provided in Appendix \ref{sec:app_scenario_background} for completeness.

We introduce a special class of XOS functions called antichain functions and summarize relevant facts about them.
Formal proofs of all claims related to antichain functions are implied by results in Appendix \ref{sec:app_scenario_antichain} for a more general class of set functions.
\begin{definition}[Antichain Set Function]
    A set function $f:2^{[n]}\to \R_{\ge 0}$ is an antichain function if there are weights $w_1,\dots,w_n \ge 0$ and laminar sets $T_1,\dots,T_n\subseteq T$ such that on any input $S\in 2^{[n]}$ we have:
    \begin{align*}
        f(S) = \max_{A\in \ac(S)} w(A)
    \end{align*}
    where $\ac(S) = \{A\subseteq S\mid T_i \cap T_j = \emptyset \text{ for all }i,j\in A \text{ with }i\neq j\}$ is the collection containing all subsets of $S$ that only contain incomparable items.
\end{definition}

It turns out that every antichain function is XOS, monotone, and normalized.

\begin{fact}[XOS, Monotonicity, and Normalized]
    If $f$ is an antichain function, then $f$ is XOS, monotone, and normalized.
\end{fact}

The last fact we require is that antichain functions admit efficient valuation, demand, and supporting price oracles.
Each of these oracles reduces to finding a maximum-weight antichain on a vertex-weighted forest.
We specifically consider forests induced by laminar type sets $\{T_i\}_{i\in S}$; these graphs have vertex set $S$ and require that every $i,j\in S$ with $T_i \subset T_j$ has $j$ as an ancestor of $i$.
Note that if two nodes have the same type set ($T_i = T_j$), $i$ and $j$ must be a parent and child, in any order. 
Such forests can be computed efficiently.

In the case of value and supporting price queries on an input subset of items $S\in 2^{[n]}$, the maximum weight antichain on the forest induced by laminar types $\{T_{i}\}_{i\in S}$ with vertex weights $w$ encodes a maximizing XOS clause.
For demand queries, the bundle of items in demand is a maximum weight antichain on the forest induced by laminar types $T_1,\dots, T_n$ with vertex weights $w_i-p_i$ for each vertex $i$.

We use dynamic programming to compute maximum weight antichains.
There is an entry $\opt[i]$ in our dynamic programming table for each node $i$.
The value of entry $\opt[i]$ is the maximum weight of an antichain in the subtree rooted at $i$.
There are two cases we consider in our recursion at node $i$.
First, the maximum weight antichain is $\{i\}$ (note that the inclusion of any descendant of $i$ breaks the antichain property).
Second, the maximum weight antichain is not $\{i\}$ but some collection of descendants of $i$.
This yields the following recursive definition:
\begin{align*}
    \opt[i] = \max\left\{w_i, \sum_{j\in \child(i)} \opt[j]\right\}
\end{align*}
where $\child(i)$ is the immediate children node of $i$ in the specified forest.
Finally, backtracking recovers the maximum weight antichain in the vertex-weighted forest in polynomial time.

\begin{fact}[Efficient Oracles]
    An antichain function admits polynomial-time implementations of value, demand, and supporting price oracles.
\end{fact}

Having established the required background, we turn our attention to the Combinatorial Auction problem and our reduction from \mwlm.


\begin{lemma}[Reduction to Combinatorial Auctions]\label{lem:scenarios_redcombo}
    For any $\alpha\in(0,1]$ and $\epsilon\ge 0$, a polynomial time $(\alpha,\epsilon)$-approximation algorithm for the Combinatorial Auctions problem implies a polynomial time $(\alpha,\epsilon)$-approximation algorithm for the Maximum Weight Laminar Matching problem.
\end{lemma}

\begin{proof}
    Given an instance of MWLM $\inner{G=(L\cup R), w, T}$, we construct an instance of Combinatorial Auctions $\inner{n, m, \{v_i\}_{i\in[n]}}$.
    Each left node corresponds with a buyer $[n] = L$.
    Each right node has an associated item $[m] = R$.
    The valuation $v_i$ of buyer $i$ for a bundle $S\subseteq [m]$ of goods is the maximum weight subset of right nodes with distinct types.
    For a bundle $S$, recall that $\ac(S) = \{ A\subseteq S\mid  T_j \cap T_{j'} = \emptyset \text{ for all distinct } j,j'\in A\}$ is the collection of subsets containing at most one item of each type.
    \begin{align*}
        v_i(S) = \max_{A \in \ac(S)} \sum_{j\in A} w_{i,j}
    \end{align*}
    Notice that $v_i$ is an antichain function.
    Thus, it admits a linear time value, demand, and supporting price oracle.

    We now create a surjective map from allocations to laminar matchings.
    Allocation $S_1,\dots,S_n\subseteq [m]$ maps to the matching in which left node $i$ has an edge to right node $j$ when the supporting prices of $i$ on allocation $S_i$ has non-zero price on item $j$:
    \begin{align*}
        M = \left\{(i,j)\in E \,\middle|\, j \in A \text{ for antichain } A = \arg\max_{A'\in\ac(S_i)} \sum_{j\in A'} w_{i,j}\right\}
    \end{align*}
    where the argmax operation breaks ties arbitrarily, but consistently.
    The mapping is surjective.
    For an arbitrary laminar matching $M$, the allocation $S_1,\dots,S_n$ where $S_i = \{j\in R\mid (i,j)\in M\}$ is associated with $M$.

    Finally, we show that the value of an allocation is equal to the weight of its associated laminar matching.
    Let $A_i \in \arg\max_{A\in\ac(S_i)} \sum_{j\in A} w_{i,j}$ be the set of items allocated to buyer $i$ with non-zero supporting price for their allocation $S_i$ (break ties arbitrarily but consistently). 
    \begin{align*}
        \sum_{i} v_i(S_i)
        = \sum_{i} \sum_{j\in A_i} w_{i,j}
        = \sum_{(i,j)\in M} w_{i,j} 
        = w(M)
    \end{align*}

    The reduction requires polynomial time.
    A polynomial time $(\alpha,\epsilon)$-approximation for combinatorial auctions implies a polynomial time $(\alpha,\epsilon)$-approximation for the MWLM Problem.
\end{proof}

\subsection{Beyond Theorem \ref{thm:probalisitc-main}:  Applications and Generalizations}
Since combinatorial auction with XOS has been shown to admit  an efficient $(1-1/e)$-approximation \cite{dobzinski2010auctions}, a straightforward application of Theorem \ref{thm:probalisitc-main} is that both keychain problem with explicitly described Probabilistic Scenarios and MWLM   admit efficient $ (1-1/e)$-approximations as well. Next we illustrate how the established reductions above can be further applied to derive novel results for classic problems. 

\subsubsection{An Application of the Reductions to Weighted Online Bipartite Matching.} 
 
To showcase   applications of our techniques beyond Keychain Problems,  we illustrate how the formal connections we established could be employed to reduce Philosopher Inequalities for Stochastic Edge-Weighted Online Bipartite Matching to the XOS Combinatorial Auction problem and, consequently, obtain a novel $(1-1/e)$-approximation algorithm for the Philosopher Inequality benchmark \cite{papadimitriou2021online}.  

\paragraph{  Edge-Weighted Online Bipartite Matching and the Philosopher Inequality.} 
In this setting, there are $n$ offline nodes and $m$ online nodes.
Before matching, the matcher is given a prior $p\in\Delta(\mathbb{R}_{\ge 0}^{n\times m})$ over edge weights. The matching process begins with the environment privately drawing edge weights $w\sim p$.
The right nodes $j=1,\ldots,m$ arrive one at a time, and upon arrival the environment reveals edge weights $w_{1,j},\dots,w_{n,j}$.
Upon observing the weights of online node $j$, the matcher must irrevocably select an offline neighbor $i$ to match to node $j$ or opt to match nothing.
The matcher must maintain a matching in which offline node $i$ has at most one neighbor. The goal is to compute a policy to compete against a Philosopher who uses their unconstrained computational power to find a Bayes optimal policy. 
Notably, our results apply to a generalized setting that allows arbitrary edge weight correlation in the prior as well as one-sided $b$-matching, whereas \cite{braverman2022max, braverman2025new, papadimitriou2021online} assume edge weights are independently drawn and restricted to standard bipartite matching. 
However, our results require that there are only polynomially-many realizable edge sets.
Although the current state of the art is a $0.678$-approximation ($>1-1/e$) for their setting, their results are not applicable here.

As a direct application of Theorem \ref{thm:probalisitc-main} and its proof techniques, we can obtain the following new results for the Philosopher Inequalities. For completeness, detailed proof of the proposition is provided in Appendix \ref{subsec:app_scenarios_obm}.

\begin{restatable}[Online Bipartite Matching Approximation]{proposition}{obm}\label{prop:scenarios_obmapprox}
    The Philosopher Inequality for the Edge-Weighted Online Bipartite Matching problem admits a $(1-1/e)$-approximation algorithm that runs in time $O(n,m,|\supp(p)|)$.
\end{restatable}



\subsubsection{Generalization of Theorem \ref{thm:probalisitc-main} to Black-box Prior Access.} Our proof techniques of Theorem \ref{thm:probalisitc-main} can be further generalized to the setting where the locksmith does not know the prior $p$  but merely has sample access to it. 
Intuitively, our algorithm can sample scenarios to estimate the reward of selecting a key at any information set within a small additive error with high probability.
By applying Lemmas \ref{lem:scenarios_redmwlm} and \ref{lem:scenarios_redcombo}, we can derive  a  $(1-1/e, \epsilon)$-approximation with small and controllable additive error and failure probability. For completeness, the detailed proof is provided in Appendix \ref{sec_app:black-box-prop}. 

\begin{restatable}
{proposition}{scenariosampleapprox}\label{prop:scenarios_sampleapprox}
In the black-box prior setting, there is a $O(\poly(\epsilon^{-1},\delta^{-1},n,m,|\calS|))$ time algorithm that computes a  $\left(1 - 1/e, \epsilon\right)$-approximate policy for the Probabilistic Scenarios setting with probability at least $1-\delta$ for any $\epsilon > 0$ and $\delta > 0$.
    This algorithm uses $O\left(\left(\frac{m^2 |\calO|^2}{\epsilon^2}\right) \ln\left(\frac{n |\calO|}{\delta}\right) \right)$ scenarios sampled i.i.d. from the prior $p$.
\end{restatable}



\subsection{Keychain Problems with Probabilistic Scenarios Are Hard}\label{subsec:scenarios_hardness}
The reductions in Theorem \ref{thm:probalisitc-main} only show that the keychain problem with probabilistic scenarios is ``no harder'' than combinatorial auctions with XOS, but do not rule out the possibility that the Keychain Problem  may be polynomially solvable. We conclude this section by ruling out that possibility with an APX-hardness result. 

\begin{theorem}[Probabilistic Scenarios Hardness of Approximation]\label{thm:scenarios_hardness}
    Assuming $P \neq NP$, there is no polynomial time algorithm for the Probabilistic Scenarios problem that guarantees a $\left(\frac{4063}{4064} + \delta\right)$-approximation, for any $\delta > 0$.
\end{theorem}

Our reduction is from a special case of 3-SAT called Max-$(3, 2B)$-SAT \cite{berman2004approximation}. 
In this special case of 3-SAT, each literal appears exactly twice among all the clauses, and each clause contains three distinct literals. We defer full details of this reduction to Appendix \ref{sec_app_prob_hard}.

Interestingly, our reduction can be applied to show an impossibility result for the Philosopher Inequality as well, as stated below and formally proved in Appendix \ref{sec_app_philosopher_hard}.   

\begin{restatable}[Online Bipartite Matching Approximation]{proposition}{obmhard}\label{prop:scenarios_obmhard}
    For the Philosopher Inequality Edge-Weighted Online Bipartite Matching problem, there is no polynomial time $\left(\frac{4063}{4064} +\epsilon\right)$-approximation algorithm for any $\epsilon > 0$ assuming that $P\neq NP$.
\end{restatable}

\section{Keychain Problems with Optimizable Chain Orders}\label{sec:order}

  \Cref{sec:scenarios} investigates the Bayesian setting in which the player  has no control over the keychains' order but only has some prior knowledge about it.
In this section, we strengthen the locksmith and allow her to choose the order of keychains.
We study how this will change the problem's structural properties and complexity. 
In this setting, the locksmith receives as input a set of keys $[n]$, a set of  keychains $\mathcal{C} := C_1, \dots, C_m$ where $C_i \in 2^{[n]}\setminus\{\emptyset\}$, and a prior $p \in \Delta[n]$ over the correct key. 
At each round $t$ of the $m$ total rounds, the locksmith chooses any previously unselected keychain $C_t$ and selects a key $k_t \in C_t$, obtaining a reward of $\textbf{1}\{k_t = k^*\}$.



Before proceeding to our algorithmic analysis, we first observe that, although  the locksmith now is allowed to adaptively 
choose the keychain order, there remains no benefit in being adaptive when 
the Keychain Problem has only one correct key.  To see this, note that
the only steps of the interaction protocol where the locksmith
faces a non-trivial decision are those in which all of the keys tested in previous rounds were found
to be incorrect keys. Given any (potentially adaptive) policy that is exploitative, 
one can transform it into an exploitative policy that chooses the order of keychains
non-adaptively by simulating the behavior of the given policy when every key tested 
turns out to be incorrect. The order in which the policy selects keychains in this special 
case can be used as its keychain ordering in all cases, without changing its expected reward.
Before finding a correct key, the modified policy's behavior is identical to that of the 
original, and after finding a correct key both policies get the same reward because they
are both exploitative. Henceforth, we will assume  without loss of generality that policies choose the order of keychains non-adaptively.

\subsection{NP-Hardness of Keychain Order Selection} \label{subsubsec:order_hardness}

When each keychain can contain an arbitrary number of keys and the locksmith can pick the order of the keychains, we show that it is NP-Hard to determine the optimal strategy for the locksmith. 
This result reinforces connections between the Keychain Problem and problems on bipartite graphs -- indeed, the Keychain Order Selection Setting is closely related to the NP-hard task of finding permutations on bipartite graphs that maximize the number of ``downward'' edges.
We formalize the connection in the following.

Let $G$ be a directed bipartite graph $G = (L \cup R, E)$ such that edges only go from $L$ to $R$. Let $\sigma_L$ and $\sigma_R$ denote orderings over the left and right nodes, respectively. We say an edge $(l, r)$ points upwards if $\sigma_L(l) > \sigma_R(r)$, and otherwise it points downwards.

\begin{definition} [Downwards-Facing Bipartite Permutation (DFBP) \cite{fertin2015obtaining}]
    Given a balanced bipartite graph $G = (L \cup R, E)$ with edges going only from $L$ to $R$, determine if there are orderings of the nodes $\sigma_L$ and $\sigma_R$ such that no edges point upwards.
\end{definition}

It is well-known that the DFBP problem is NP-hard \cite{fertin2015obtaining,gerbner2016topological}.\footnote{It turns out that reordering the left-side and right-side nodes of a bipartite graph is equivalent to independently permuting the rows and columns of its adjacency matrix. The goal of obtaining a downwards facing bipartite graph then corresponds   to making the adjacency matrix an \textit{upper triangular} matrix (i.e.,   entries below the main diagonal are all $0$'s).}  
Our main result of this section is to show that DFBP reduces to the Keychain Order Selection problem, implying its NP-hardness.  


\begin{restatable}{theorem}{keychainorderhard}
\label{thm:order_hard}
The  Downwards-Facing Bipartite Permutation (DFBP) problem reduces in polynomial time to  a subclass of the Keychain Order Selection problem that has an equal number of keys and keychains ($n=m$) and a uniform prior over key correctness. 
\end{restatable}
 

A natural first attempt is to directly interpret the given DFBP instance as a Keychain Order Selection instance, mapping left nodes to keys and right nodes to keychains.
Unfortunately, this correspondence fails because the Keychain objective depends on more than just the number of downwards edges.
There are orderings of left and right nodes of YES instances of DFBP that would not be optimal for the corresponding Keychain Instance (and may not even achieve positive utility). 

To overcome this obstacle, we establish structural characterizations of the optimal solution in the uniform-prior Keychain setting.
When the prior distribution is uniform, the locksmith's reward is at most $(n+1)/2$, achievable if and only if the nodes in the bipartite graph corresponding to the Keychain Instance can be ordered such that there are $n(n+1)/2$ downwards pointing edges (Appendix \ref{subsec:app_observations_for_hardness}).

Armed with these observations, we construct a modified reduction.
Given an instance of DFBP, we augment the graph with an additional left and right node and take the transpose of the complement of the resulting graph.
This serves as our instance of the Keychain Order Selection problem, and we show that the constructed instance achieves utility $(n+1)/2$ if and only if the original DFBP instance is a YES instance, from which we conclude that the Keychain Order Selection problem is NP-hard. 
The formal proof of \Cref{thm:order_hard} appears in Appendix \ref{sec:app_order_hardness_proof}.

\subsection{A Simple (1/2)-Approximation} \label{subsubsec:order_half_approx}
We complement the above hardness with a simple algorithm that achieves half of the optimal value. The algorithm is extremely simple, though we conjecture it may be optimal (see a more detailed discussion in Appendix \ref{subsec:order_opendirections}). Specifically, the algorithm simply solves the basic Keychain Problem of Section \ref{subsec:intro_warmup} on any ordering $\sigma$ over the keychains  as well as on its reverse $-\sigma$. We show that picking the better among these two policies yields a $(1/2)$-approximation. 


\begin{restatable}{proposition}{orderhalfapprox}
    \label{thm:order_approx}
    There is a $(1/2)$-approximation of the Bayes optimal strategy for the Keychain Order Selection problem.
\end{restatable}



A key insight from our proof is that the tasks of selecting the order of the keychains and selecting which key to play on each keychain can be decoupled.
Suppose the locksmith is constrained to play a fixed \textit{key selection} that a unique $k_c$ to each keychain $c$, regardless of the order.
For any ordering of the keychains and any key $k$, at least half of the occurrences of $k$ must appear before or after the key is played according to to the key selection function.
So, the combined reward of any ordering and its reverse serves as an upper bound for the best ordering for any key selection function, implying that one of the two at least half the value of the optimal ordering for this fixed key selection function.
Finally, the optimal policy can be mapped to such a key selection function, completing the argument.
The detailed proof is in Appendix \ref{sec:app_order_approx}.

Our best-of-two approximation algorithm is strikingly similar to the folklore best-of-two (1/2)-approximation of Maximum Acyclic Subgraph (MAS), in which any ordering or its reverse is a (1/2)-approximation of the optimal solution.
These similarities suggest a deeper connection to MAS, but finding an approximation preserving hardness reduction, if one exists, is a challenging problem that we leave to future work.    
Despite their similarities, the Keychain Order Selection setting includes an additional constraint that played keys must exist on the current keychain.
This constraint is difficult to encode in the objective of MAS and prevents a straightforward application of the known hardness of approximation results for MAS \cite{guruswami2008beating, bhangale2019ug}.
Further discussion of the constraint is presented in \Cref{subsec:order_opendirections}.


\section{Concluding Discussions and Future Directions} 
In this paper, we introduced a class  of problems, collectively termed Keychain Problems, which naturally captures the opportunity cost minimization problem in sequential decision making. Employing polynomial-time reductions, we formally establish the connections between variants of Keychain problems and various problems pertaining to bipartite graphs, ranging from the basic maximum weight bipartite matching, to the novel maximum weight laminar matching (which then reduces to combinatorial auctions with XOS valuation functions), and downwards-facing Bipartite Permutation. These formal connections not only deepen our understanding about the structure and richness of the Keychain Problem, but also can be applied to derive new algorithmic results such as the $(1-1/e)$ approximation for the Philosopher Inequality in some natural settings. 

Our paper leaves open many rich areas of exploration.
In the Probabilistic Scenarios setting, our reductions run in time polynomial in the number of scenarios, meaning that they will not work for settings where we have succinct priors over exponentially many scenarios (e.g. if the next keychain is chosen uniformly at random from the remaining keychains).
In the Order Selection setting, one may hope to strengthen the connections to ordering CSPs to preserve approximation-hardness, and one can ask if our $(1/2)$-approximation is tight.
Finally, we present preliminary APX-hardness results in the Many Correct Keys setting, but there is still a wealth of unanswered and interesting research questions in this area.
A detailed discussion of open directions is presented in Appendix \ref{sec:questions}.


\bibliographystyle{plain}
\bibliography{references}
\newpage 
\appendix
\section*{\Large Appendix: Table of Contents}
\vspace{2mm}

\noindent
\textbf{\ref{sec:questions}. Open Questions and Future Directions}\dotfill \pageref{sec:questions}

\noindent
\textbf{\ref{sec:app_intro_example}. Advisor Search Calculations}\dotfill \pageref{sec:app_intro_example}

\noindent
\textbf{\ref{sec:app_related_works}. Additional Discussions on Related Works}\dotfill \pageref{sec:app_related_works}

\noindent
\textbf{\ref{sec:app_intro_policy}. Omitted Discussions and Proofs from Section \ref{sec:warm-up}}\dotfill \pageref{sec:app_intro_policy}

\noindent
\textbf{\ref{sec:app_scenario}. Omitted Details and Proofs from Section \ref{sec:scenarios}}\dotfill \pageref{sec:app_scenario}

\noindent
\textbf{\ref{sec:app_order}. Omitted Details and Proofs from Section \ref{sec:order}}\dotfill \pageref{sec:app_order}

\noindent
\textbf{\ref{subsec:app_scenarios_obm}. Weighted Online Bipartite b-Matching Algorithms}\dotfill \pageref{subsec:app_scenarios_obm}

\noindent
\textbf{\ref{sec:app_ai_disclosure}. AI Software Disclosure}\dotfill \pageref{sec:app_ai_disclosure}

\newpage
\section{Open Questions and Future Directions}\label{sec:questions}
\subsection{Approximation with Exponentially Many Probabilistic Scenarios}\label{subsec:scenarios_openprob}
Our algorithms for the Probabilistic Scenarios setting (\Cref{sec:scenarios}) run polynomially in the number of scenarios.
We ask the intriguing question: Is there an efficient $(1-1/e)$-approximation for instances with succinct priors over exponentially many scenarios?
Specifically, we consider prior distributions that admit an efficient future scenario sampling oracle.
This oracle takes as input an information set $o$ and returns a scenario $s$ sampled from the prior conditioned on the occurrence of $o$: scenario $s$ is returned with probability $(p_s/p_o)\cdot\1{s\in C(o)}$.

\begin{question}[Exponentially Many Scenarios]
    Given a future scenario sampling oracle, is there a polynomial-time $(1-1/e)$-approximation algorithm for the probabilistic scenarios problem?
\end{question}

The main observation giving hope that such an algorithm is possible is the fact that the locksmith does not need an entire description of her policy $\pi$.
She only needs the output of her policy $\pi$ restricted to the set of realized information sets, which is observed in an online fashion.
We foresee the future scenario sampling oracle being applied to each observed information set to learn which available key should be tested.

In addition to the instance in which each of the $m$ keychains is drawn i.i.d.\ from a known distribution (analogous to the online bipartite matching setting investigated in \cite{papadimitriou2021online, braverman2022max,braverman2025new}), this future sampling algorithm could be applied to much more intricate stochastic models.
The most natural of these instances has a prior that presents a sequence of known keychains whose ordering is chosen uniformly at random.

\subsection{Open Directions in the Keychain Order Selection Setting} \label{subsec:order_opendirections}

Both the approximation algorithm and the NP-Hardness result in Section \ref{sec:order} point towards connections between the Keychain Order Selection problem and quadratic assignment problems \cite{nagarajan2009quadratic}. 
Indeed, the (1/2)-approximation is very similar to the known folklore (1/2)-approximation of Maximum Acyclic Subgraph (MAS), in which either an ordering of the nodes or its reverse satisfies at least 1/2 the edges.
Similarly, the hardness reduction from the UTMP problem suggests a deeper connection to quadratic assignment problems in general.

However, an approximation-preserving connection between the two problems remains elusive: in UTMP, the objective sums up the number of 1 entries in the upper triangle.
The Keychain Order Selection problem has \textit{almost} the same objective except that it only counts 1 entries in rows of the upper triangle such that the main diagonal entry of that row is a 1 ($A_{i, i} = 1$).
We refer to this as the diagonal constraint.

The diagonal constraint is not easily encoded in either MAS or the UTMP objectives, whose focus is simply to maximize the number of ``satisfied'' (in some way) edges.
A high value solution to MAS or UTMP can have potentially zero value in the Keychain Order Selection problem if all the diagonal entries are 0.
Conversely, a high value solution for the Keychain Order Selection problem may not be the highest value solution in MAS or UTMP (although it \textit{will} retain the same value), as it may be dwarfed by a solution that disregards the diagonal constraint entirely.

We leave open multiple directions for the Keychain Order Selection problem.
If the diagonal constraint issue can be navigated properly, then steps could be made towards a proper hardness of approximation result.
It would be interesting to learn whether known hardness results in MAS \cite{guruswami2008beating} generalize to the Keychain Order Selection problem. 

\begin{question}[Hardness of Approximation]
    Assuming standard complexity theory assumptions ($P \neq NP$ or the Unique Games Conjecture), is it hard to approximate the Keychain Order Selection problem past 1/2?
\end{question}

This would imply a deep and concrete connection to MAS and showing a hardness of approximation past 1/2 would imply that our simple algorithm  is optimal.
However, a direct approximation-preserving reduction from MAS seems difficult as it is unclear how to construct a Keychain Order Selection instance that both mirrors a given MAS instance and enforces the diagonal constraint in the objective.
On the other hand, and perhaps more interestingly: if additional structure of the Keychain Order Selection problem can be leveraged properly, it may be possible to break the $\frac12$OPT barrier.

Another promising direction is in developing a PTAS when the problem is \textit{dense}: settings in which the value of the optimal solution is guaranteed to be at least $c \cdot m$ for some constant $c > 0$.
Loosely speaking, it is easy to identify ``good'' keys to try in these settings: either each relevant key shows up many\footnote{By many here, we mean roughly $\Omega(m)$ times.} times (opportunity cost of skipping is low), or the key has a high probability of opening the lock (the opportunity cost of skipping is recognizably high).

\begin{question}[PTAS for Dense Instances]
    Does the Keychain Order Selection problem admit an additive or multiplicative PTAS for dense instances?
\end{question}

One possible approach for dense settings is to use the techniques of \cite{arora1996new}, which could possibly be generalized to incorporate the conditional nature of the diagonal constraint.
However, a direct application fails as the diagonal constraint eventually manifests as a linear constraint that needs to be satisfied \textit{exactly} in any solution, but \cite{arora1996new} only satisfies linear constraints approximately.

\subsection{Known Keychain Order and Many Correct Keys}\label{subsec:multi_openpoly}
There are a wealth of enticing directions moving forward in the Multiple Correct Key setting with known keychain order.
We offer two natural settings which are largely unexplored.

In the first setting, the prior on correct keys has polynomial size support.
Much like small priors in the Probabilistic Scenarios setting facilitated the design of efficient algorithms, we question if small priors in the Multiple Correct Key setting allows that design of similar efficient algorithms.
\begin{question}[Small Priors]
    Is the Multiple Correct Key setting with a prior on correct keys with small support tractable or not?
\end{question}

Alternatively, the reader can ask whether our intractability result is truly contingent on the event that keys opening the lock are mutually dependent.
In other words, we suggest looking into the Multiple Correct Key setting when the correctness of keys is mutually independent.

\begin{question}[Independent Key Acceptances]
    Is the Multiple Correct Key setting with independent key acceptances tractable?
\end{question}

\section{Advisor Search Calculations}\label{sec:app_intro_example}
See Section \ref{subsec:intro_example} for a description of the Advisor Search Example.
We introduce some notation to facilitate our analysis.
The events $A$, $B$, and $C$ denote when Alice, Bob, and Carol would be a good advisor for Lily, respectively.
The random variable $N_A(t)$ denotes the number of times Alice is available on rounds $t\in T$.
We use the notation $N_B(t)$ and $N_C(t)$ in the same way, but for Bob and Carol, respectively.
Finally, $S$ denotes the event that Alice takes a sabbatical period during the second rotation.

The greedy policy selects Alice for the first rotation.
She yields the greatest expected number of good rotations.
Specifically, if Alice is selected on the first round, we can expect to have one successful rotation with her:
\begin{align*}
    \Pr[A]\cdot \E\left[N_A(1,2,3)\right]
    = \left(\frac{3}{7}\right)\left( \frac{2}{3} \cdot 2 + \frac{1}{3} \cdot 3 \right)
    = 1
\end{align*}
In contrast, selecting Bob or Carol for the first rotation yields a lower number of expected successful rotations:
\begin{align*}
    \Pr[B]\cdot\E[N_B(1,2,3)] = \Pr[C]\cdot\E[N_C(1,2,3)] = \frac{6}{7}
\end{align*}
If collaboration with Alice is successful during the first rotation, then Lily should continue to collaborate with Alice on subsequent rotations.
Otherwise, she should pick Bob on the second rotation and Carol on the third rotation, if collaboration with Bob is unsuccessful.
The expected number of successful rotations under Lily's greedy policy is:
\begin{align*}
    &\E[\text{number of successful rotations of greedy}] \\
    &= \Pr[A] \cdot \E[N_A(1,2,3)] + \Pr[B] \cdot \E[N_B(2,3)] + \Pr[C] \cdot \E[N_C(3)] \\
    &= \left(\frac{3}{7}\right)\left(\frac{7}{3}\right) + \left(\frac{2}{7}\right)\left(2\right) + \left(\frac{2}{7}\right)\left(1\right) 
    = \frac{13}{7}
\end{align*}

We now analyze the second optimal policy.
To start, we consider the case in which Alice is on sabbatical.
If Lily finds the correct advisor, she should continue to work with them when possible.
Lily will first choose to collaborate with Bob, and Carol second if Bob is not the right advisor. 
If neither Bob nor Carol is are good advising fit, then Lily will choose to collaborate with Alice for the third rotation.
Lily's expected number of successful rotations conditioned on Alice's sabbatical is
\begin{align*}
    &\E[\text{successful rotations of optimal} \mid S] \\
    &= \Pr[A\mid S]\cdot\E[N_A(3) \mid S] + \Pr[B\mid S]\cdot\E[N_B(1,2,3)\mid S] \\
    &\quad\quad + \Pr[C\mid S]\cdot\E[N_C(2,3)\mid S] \\
    &= \left(\frac{3}{7}\right) \cdot 1 + \left(\frac{2}{7}\right) \cdot 3 + \left(\frac{2}{7}\right) \cdot 2
    = \frac{13}{7}
\end{align*}
Now we analyze the case when Alice does not go on sabbatical and is available for every rotation.
Lily chooses to collaborate with Bob for the first rotation.
On the second rotation, she should choose to work with Alice if Bob is not a good advising fit.
Finally, if neither Alice nor Bob is a good advising fit, she should collaborate with Carol for the last rotation.
Lily's expected number of successful rotations in this event is:
\begin{align*}
    &\E[\text{successful rotations of optimal} \mid \bar{S}] \\
    &= \Pr[A\mid \bar{S}]\cdot\E[N_A(2,3) \mid \bar{S}] + \Pr[B\mid \bar{S}]\cdot\E[N_B(1,2,3)\mid \bar{S}] \\
    &\quad\quad+ \Pr[C\mid \bar{S}]\cdot\E[N_C(3)\mid \bar{S}] \\
    &= \left(\frac{3}{7}\right) \cdot 2 + \left(\frac{2}{7}\right) \cdot 3 + \left(\frac{2}{7}\right) \cdot 1
    = \frac{14}{7}
\end{align*}
Applying the law of total expectation, we find that Lily's expected number of successful rotations for this second policy is:
\begin{align*}
    &\Pr[S] \cdot \E[\text{successful rotations of optimal} \mid S] \\
    &\quad\quad + \Pr[\bar{S}]\cdot \E[\text{successful rotations of optimal} \mid \bar{S}] \\
    &= \left(\frac{2}{3}\right) \left(\frac{13}{7}\right) + \left(\frac{1}{3}\right)\left(\frac{14}{7}\right)
    = \frac{40}{21} > \frac{13}{7}
\end{align*}

\section{Additional Discussions on Related Works}\label{sec:app_related_works}
Relevant works to our combinatorial auction reduction in the Keychain Setting with Probabilistic Scenarios are \cite{dobzinski2010auctions} and \cite{feige2006welfare}. They provide similar $(1-1/e)$-approximations that involve solving a configuration linear program and rounding ``preallocated'' items using supporting price queries.
We invoke this algorithm in Lemma \ref{lem:app_scenarios_combo} in a black-box manner to recover a $(1-1/e)$-approximation for the Probabilistic Scenarios setting.
Preceding the previously mentioned works is \cite{dobzinski2005auctions}, which provides a $(1/2)$-approximation.
Other works in the rich space of combinatorial auctions study different valuation functions \cite{khot2005inapproximability,lehmann2001combinatorial,schrijver2000combinatorial}, truthful mechanisms \cite{archer2004approximate,assadi2020improved,bartal2003incentive,dobzinski2011impossibility,lehmann2002truth}, and more \cite{sandholm2002algorithm}.
Finally, we note that the work by \cite{cetin2025online} provides an algorithm for online stochastic optimization with arbitrary correlations. However, their results do not apply to our scenarios setting, as their algorithm would require a constant bound on the number of keychains that any one key appears on.

The order selection setting bears similarities to ordering problems such as Maximum Acyclic Subgraph (MAS) and the maximum Quadratic Assignment Problem (maxQAP).
Maximum Acyclic Subgraph is a classic NP-Hard optimization problem \cite{karp1972reducibility}, with a folklore $(1/2)$-approximation and a $(2/3)$-inapproximability result assuming $P \neq NP$ \cite{austrin2015np, bhangale2019ug, newman2000approximating}.
Under the Unique Games Conjecture, \cite{guruswami2008beating} shows that even beating a $(1/2)$-approximation is intractable.
In contrast to MAS, the Order Selection setting has the locksmith select two coupled orderings -- one over keychains and another over keys encoding when she would like to play a untested key -- so hardness of approximation results for MAS do not directly apply.

Optimizing over pairs of orderings suggests links to maxQAPs.
We establish a concrete connection between the order selection setting and the NP-complete Upper Triangular Matrix Permutation (UTMP) problem \cite{fertin2015obtaining, gerbner2016topological}.
More generally, QAPs were introduced in \cite{koopmans1957assignment} and further works \cite{makarychev2014maximum, nagarajan2009quadratic} have studied various forms of maxQAPs. \cite{nagarajan2009quadratic} obtains an approximation under the triangle inequality, and \cite{makarychev2014maximum} finds inapproximability results for general maxQAPs.
Despite our connection to UTMP, the pair of orderings in the Order Selection setting must satisfy an additional constraint: keys must be on the current chain when played.
It is unclear how to encode this constraint into the objective or constraints of both UTMP and maxQAP.

\section{Omitted Discussions and Proofs from Section \ref{sec:warm-up} }\label{sec:app_intro_policy}

\subsection{Proof of Observation \ref{ob1:exploitative}}\label{sec:app_obs1}
\subsubsection*{Optimality of Exploitation}
  Consider an arbitrary policy $\pi$ which is not exploitative.
    There is a non-zero probability event $E$ in which $\pi$ will select $k_t \neq k^*$ despite having identified $k^*$ and being faced with a keychain with $k^*\in C_t$.
    Construct a new policy $\pi$ that selects $k^*$ whenever event $E$ occurs.
    This algorithm yields equally as much reward when event $E$ does not occur.
    Moreover, it recovers strictly greater reward when event $E$ does occur, thus it yields greater expected reward.

\subsubsection*{The Keychain Problem as an MDP and Optimality of Deterministic Policy}
 Here we show how to represent the  key setting with one correct key and known keychain order as a  Markov Decision Process (MDP), which thus admits an optimal policy that is deterministic. The   variants of probabilistic scenarios  can be argued similarly.

There are two types of states in the constructed MDP. The first collection of states encodes when the locksmith has not identified the correct key.
These states are tuples $(U,t)$, where $U$ is a set of untested keys and $t$ is a round.
Additionally, there is a terminal state $\textsc{Stop}$ to model when the locksmith has identified the correct key.

Consider if the locksmith is in state $(U,t)$.
Her available actions are any keys on the current chain $k\in C_t \cup \{\perp\}$.
Consider if she plays an untested key $k_t \in U$.
Then, with probability $\Pr_{k^*\sim p}[k_t = k^*\mid k^*\in U]$, she moves to the terminal state $\textsc{Stop}$ and receives reward for all future appearances of $k_t$: $\sum_{\tau \ge t} \1{k_t\in C_t}$.
With the remaining probability, she moves to state $(U \setminus \{k_t\}, t+1)$ and obtains no reward.
Finally, in the case that she selects an already tested key (or the null key $\perp$), she deterministically moves to state $(U, t+1)$ and receives no reward.

\subsection{Proof of Proposition \ref{prop:intro_fixedalgo}}\label{append_sec_reduction_mwbp}
 Let $  \langle n, \{ C_t \}_{t=1}^m, p \rangle$ denote the keychain problem instance, and $\pi$ be an arbitrary deterministic and exploitative policy. Construct a bipartite graph with left-side nodes as all keys and right-side nodes as all keychains. An edge $e=(k, t)$ exists if and only if key $k \in C_t$.  
If policy $\pi$ selects key $k$ on keychain $C_t$, the edge $e=(k, t)$ must exist; moreover, its future expected reward is: 
\begin{align*}
    r_{k,t} = \underbrace{\Pr_{k^*\sim p}[k^* = k]}_{\text{Probability of exploitation}}\underbrace{\left(\1{k\in C_t} \sum_{\tau=t}^{m} \1{k\in C_{\tau}}\right)}_{\text{Future reward of exploitation}}
\end{align*}
Hence, the (expected) reward of  policy $\pi$ is:
\begin{align*}
    r(\pi) &= \E_{k^*\sim p}\left[\sum_{t} \1{k_t = k^*}\right] \\
    &= \E_{k^*\sim p}\left[\sum_{t} \1{\pi(t) = k^* \text{ and }k^* \in C_t} \sum_{\tau = t}^{m} \1{k^*\in C_{\tau}} \right] \\
    &= \sum_{k, t} \left(\left( \Pr_{k^*\sim p}[k^* = k] \cdot \1{k\in C_t} \sum_{\tau=t}^{m} \1{k\in C_\tau} \right)\1{\pi(t) = k}\right)  \\
    &= \sum_{k,t} r_{k,t} \cdot \1{\pi(t) = k}
\end{align*}

Notice that the injective function $\pi$ representation of a policy can be equivalently viewed as a matching between keys and keychains.
By setting the edge weights of this graph to the expected future reward of testing keys, we produce a bipartite graph where the weight of any matching is the expected value of its corresponding policy.

Formally, construct the bipartite graph with   left nodes denoting all keys and right nodes denoting all chains. 
There is an edge $(k, t)$ from left node $k$ to right node $t$ if and only if $k \in C_t$,
and its edge weight is set to $w_{k,t} = r_{k,t}$.  A matching $M$ corresponds to the policy $\pi$ when $\pi(t) = k$ if and only if $(k,t)\in M$.
    This correspondence is bijective.
Observe that the weight of a matching $M$ is equal to the expected reward of its counterpart policy $\pi$:
    \begin{align*}
        w(M) = \sum_{(k,t)\in M} w_{k,t} = \sum_{k,t} r_{k,t} \1{\pi(t) = k} = r(\pi)
    \end{align*}
The reduction takes $O(nm^2)$ time since each $r_{k,t}$ can be computed in $O(m)$ time.

\subsection{Sub-optimality of Exploitative Policies with Multiple Correct Keys}\label{appendix:multi_exploit_suboptimal} 
To further illustrate the challenge with many successful keys,   we give an example to show that exploitative policies may not be optimal. In our example, even when one of the correct keys is identified, the optimal policy needs to intentionally avoid playing that correct key despite its availability in the current chain.

\begin{fact} \label{fact:multi_exploit_suboptimal}
    There exist instances of the Keychain Problem with Many Correct Keys in which exploitative policies are suboptimal. This result holds even when the events that each key opens the lock are mutually independent.
\end{fact}

Before giving the explicit construction for Fact \ref{fact:multi_exploit_suboptimal},  we provide the reader with an intuition for why we might expect for such a statement to be true.
When there are multiple correct keys, there is
an opportunity cost to exploiting: namely, we lose out on the ability to try a previously untested key and reduce uncertainty in the future.
This stands in contrast to settings in which there is only one correct key, where once the correct key is identified, exploiting is best as no other key can work.
Our construction leverages this fact so that exploitative strategies suffer more losses in the future than they gain in present reward.

\begin{proof} 
We construct an instance of the Keychain Problem with Multiple Correct Keys such that playing a known correct key is provably suboptimal.

For a given positive integer $x$, there are $m = 2x+1$ keychains and $n = 3x + 2$ distinct keys.
There are $x$ pairs of keychains $(\{a_i, b_i\}, \{k_2, b_i, c_i\}), \forall i \in [x]$. 
The probability that any key of the form $a_i$ or $c_i$ is accepted is $0.51$. 
The probability that any key of the form $b_i$ is accepted is $0.5$. 
Finally, the first keychain is exactly $\{k_1, k_2\}$, with the probability that $k_1$ is accepted is $1$ and the probability that $k_2$ is accepted is $1-\epsilon$ for an arbitrarily small constant $\epsilon > 0$.
The keychains are ordered as follows: after the first keychain, all keychains of the form $\{a_i, b_i\}$ appear before the first keychain of the form $\{k_2, b_i, c_i\}$ appears. 
A visual representation of the the instance is presented in Table \ref{instance:multi_key_exploit_suboptimal}.

\begin{table}[ht]
    \centering
    \begin{tabular}{| c | c | c | c | c | c | c | c | c |}
        \hline
        Key &  $p_k$ & c1 & c2 & c3 & c4 & c5 & c6 & c7 \\
        \hline
        $k_1$ & 1         & o  &    &    &    &    &    &   \\
        \hline
        $k_2$ & 1-$\epsilon$ & o  &    &    &    & o   & o   & o \\
        \hline
        $a_1$ & 0.51      &    & o  &    &    &    &    &    \\
        \hline
        $b_1$ & 0.5       &    & o  &    &    & o  &    &    \\
        \hline
        $c_1$ & 0.51      &    &    &    &    & o  &    &    \\
        \hline
        $a_2$ & 0.51      &    &    & o  &    &    &    &    \\
        \hline
        $b_2$ & 0.5       &    &    & o  &    &    & o  &   \\
        \hline
        $c_2$ & 0.51      &    &    &    &    &    & o  &    \\
        \hline
        $a_3$ & 0.51      &    &    &    & o  &    &    &  \\
        \hline
        $b_3$ & 0.5       &    &    &    & o  &    &    & o  \\
        \hline
        $c_3$ & 0.51      &    &    &    &    &    &    & o \\
        \hline
    \end{tabular}
    \caption{An example instance of the Keychain Problem with Multiple Correct Keys where exploitation is not optimal. In this example, there $3$ pairs of keychains, as well as the first chain.}
    \label{instance:multi_key_exploit_suboptimal}
\end{table}

    Much like the advisor search example in \Cref{subsec:intro_example}, this example offers two natural policies: one that acts greedily and one that sacrifices immediate rewards to explore alternatives that offer rewards in the future.
    The first exploitative strategy favors the guaranteed reward of correct keys that have been identified.\footnote{It is worth noting that we will treat $k_1$ as an ``identified'' key because it is known to open the lock with probability $1$, although strictly speaking, it is not yet tested. This can be rectified by adding a keychain c0 at the beginning of the sequence that only contains key $k_1$.}
    On the first keychain, this strategy tests key $k_1$, which offers no future exploratory benefit.
    The second policy instead opts to test key $k_2$ first.
    By doing so, with small probability $\epsilon$, it rules out one key on each keychain of the form $\{k_2,b_i,c_i\}$.
    By further exploring the key $b_i$ on the keychains with form $\{a_i, b_i\}$, this second policy guarantees it receives a reward when possible on the last $n$ keychains. 

Consider the \textsc{Exploit} strategy that first plays $k_1$, then continues play the $a_i$ keys for the next $x$ keychains. 
Then, when the first keychain of the form $\{k_2, b_i, c_i\}$ appears, \textsc{Exploit} tries $k_2$. 
If $k_2$ is accepted, then it plays $k_2$ for all remaining keychains, otherwise it plays $c_i$ for all $x-1$ remaining keychains. 
The expected value of $\textsc{Exploit}$ is:
\begin{align}
    \mathbb{E}[r(\textsc{Exploit})] &= 1 + 0.51x + (1-\epsilon)x + \epsilon[(x-1)(0.51)] \nonumber \\
    &= 1 - 0.51\epsilon + [1.51 - 0.49 \epsilon]x \label{EV_Exploit}
\end{align}

If we pick the \textsc{OPT} strategy that first plays $k_2$, there are two options. 
With probability $1-\epsilon$, $k_2$ works and we automatically win on the 2nd keychain of each pair. 
Then, it is best to take all the $a_i$ keys when we see the first keychain of each pair. 
Otherwise, it is better to play $b_i$ for the first keychain of each pair. 
If $b_i$ works, then we automatically win on the second keychain of that pair. 
If $b_i$ doesn't work, then we play $c_i$ on the second keychain of each pair. 
The expected value of \textsc{OPT} is:
\begin{align}
    \mathbb{E}[r(\textsc{OPT})] &= (1 - \epsilon) \left[ (x+1) + 0.51x \right] + \epsilon[(0.5(2)x + 0.5(0.51)x] \nonumber \\
    &= 1 - \epsilon + [1.51 - 0.255 \epsilon]x \label{EV_Explore}
\end{align}

The advantage of \textsc{OPT} over \textsc{Exploit} is:
\begin{align*}
    \mathbb{E}[r(\textsc{OPT})] - \mathbb{E}[r(\textsc{Exploit})] &= [0.235 x- 0.49] \epsilon > 0
\end{align*}
Note that for any fixed $\epsilon > 0$, the advantage of exploration can be made arbitrarily large by taking the number of pairs $x$ to $\infty$.

All that remains to show is that for small enough $\epsilon$, \textsc{Exploit} is the best strategy that begins by playing $k_1$.
At this point, it will be easier to condition on the event that $k_2$ is correct.
The value of \textsc{Exploit} under this viewpoint is:
\begin{align}
    \mathbb{E}[r(\textsc{Exploit})] &= 1 + (1 - \epsilon)[(0.51 + 1)x] + \epsilon[(x-1)(0.51 + 0.51) + (0.51 + 0)] \nonumber \\
    &= 1 + (1-\epsilon)[1.51x] + \epsilon[1.02(x-1) + 0.51] \nonumber \\
    &= 1 + (1-\epsilon)[1.51x] + \epsilon[1.02x - 0.51] \label{cond_view_exploit}
\end{align}

For any policy $\pi$ that plays $k_1$ on the first round, it enters an instance in which it can choose to play an $a_i$ key $w$ times and the $b_i$ key the remaining $(x-w)$ times.
Supposing $(x-w) \geq 1$ (else $\pi$ would be the same policy as \textsc{Exploit}), a (loose) upper bound\footnote{Note that this is a loose upper bound -- we do not account for the fact that the locksmith must ``waste'' a round trying $k_2$ in the cases it doesn't work.} for the value of this policy is:
\begin{align}
    \mathbb{E}[r(\pi)] &\leq 1 + (1-\epsilon)[w(0.51 + 1) + (x-w)(0.5 + 1)] \\
    &\quad\quad + \epsilon[w(0.51 + 0.51) + (x-w)(0.5(2) + 0.51)] \nonumber \\
    &= 1 + (1-\epsilon)[1.51w + 1.5(x-w)] + \epsilon[1.02w + 1.51(x-w)] \label{cond_view_pi}
\end{align}

Setting the expected reward of $\pi$ (Line \ref{cond_view_pi}) less than the expected reward of \textsc{Exploit} (Line \ref{cond_view_exploit}) and solving for $\epsilon$ yields that \textsc{Exploit} is a better policy than $\pi$ when $\epsilon \leq \frac{1}{50 + \frac{51}{x-w}}$.
Since $(x-w) \geq 1$, when $\epsilon = \frac{1}{1000} < \frac{1}{50 + 51}$, \textsc{Exploit} is the best strategy that begins by selecting $k_1$.
\end{proof}

\subsection{APX-hardness with Many Correct Keys and Known Keychain Order }\label{app_subsec:many-key}
To illustrate the challenge of many correct keys that opens the lock, we exhibts a natural setting here and show that it is APX-hard to compute a Bayes optimal policy. 


The setting with many correct keys and known keychain order is the same as  described  at the beginning of  Section \ref{sec:warm-up}, except that now there is a set $K^* \subseteq [n]$ of correct keys and the locksmith gets reward $\1{k_t \in K^*}$.
To study the problem's computational complexity, we shall consider settings with a succinct polynomial-size representation of the prior distribution of correct keys.

\paragraph{Succinct prior with dueling key pairs. }To show the hardness, we consider a simple case with dueling key pairs. Similar to the standard setup, let there be $m$ keychains $C_1, \dots, C_m \subseteq [n]$.  Suppose the number of keys $n$ is even. The prior distribution about correct keys  is such that one and only one key among $2i-1$ and  $2i$ can open the lock, and $2i-1$ is the correct key with probability $p_i \in [0,1]$ (hence $2i$ is correct with probability $1-p_i$). These events are independent for different $i$'s, so there are precisely $n/2$ correct keys. An input to such an instance is thus $\langle  n, \{ C_t \}_{t=1}^m, \{ p_i \}_{i=1}^{n/2} \rangle$.  

\begin{restatable}[Many Correct Keys APX-Hardness]{proposition}{multikeyhard}\label{thm:multi_hard}
It is APX-hard to compute a Bayes optimal policy for the Keychain problem with many correct keys and a succinct prior with dueling key pairs, as described above.  
\end{restatable}

The proof is via a reduction from a special case of Vertex Cover with bounded-sized sets, described as follows.
\begin{definition} [Minimum Vertex Cover-$3$]
    Given a graph $G = (V, E)$ on which every node has degree at most $3$, compute a vertex cover of minimum size.
\end{definition}

\begin{proof}
    The idea of the reduction is to turn the edges into keychains and the nodes into keys; therefore, selecting a key from a keychain corresponds to covering an edge with a node.
    We construct an instance of the Keys and Keychains with Multiple Correct Keys and Succinct Prior Problem from an instance of Vertex Cover-3 as follows. Instantiate $m = |E|$ keychains $C_1, C_2, \dots, C_m$, one chain for each edge. Instantiate $2n$ keys ($n$ dueling key pairs); keys $(2i-1)$ and $2i$ correspond to node $i\in V$ (where $|V| = n$). For every $i \in [n]$ and $j \in [m]$, keychain $C_j$ includes keys $(2i-1)$ and $2i$ if and only if node $i$ is an endpoint of edge $j$. Finally, set $p_i = 1/2$ for all $i \in [n]$.

    First, observe that the optimal policy must play keys corresponding to nodes in a vertex cover, since it must choose one key from each keychain, which corresponds to choosing a node that covers the associated edge.
    Also note that an optimal policy always selects a key that can open the lock, except for the case when both of the keys corresponding to the selected edge haven't been tested in the previous rounds; to simplify analysis, we restrict our analysis to such policies without loss of generality.
    For any policy $\pi$, we denote $S_\pi$ as the random subset of nodes (key-pairs) that policy $\pi$ plays throughout the Keychain Game.
    The expected reward of $\pi$ is:
    \begin{align*}
            r(\pi) = m-\frac{1}{2} \E|S_\pi|
    \end{align*}
    since with probability $1/2$ independent of the previous observations, every policy will incur a loss of at least $1$ upon testing a key from an unselected pair for the first time.

    Now we show that a high reward policy for the Keychain Problem with a known order and many correct keys implies a constant fraction approximation algorithm for Vertex Cover with a high probability guarantee.
    Intuitively, a high reward policy will play key-pairs corresponding with small cardinality vertex covers in expectation.
    By repeating the sequential decision-making process for such a policy, we can sample many vertex covers, one of which will have low cardinality with high probability.

    Let $\pi^*$ be a Bayes optimal policy for the Keychain Problem, and $S^*$ be a minimum cardinality vertex cover for the Set Cover instance.
    The Bayes optimal policy will select key-pairs corresponding to a minimum vertex cover to maximize its reward:
    \begin{align*}
        r(\pi^*) = m - \frac{1}{2} |S^*|
    \end{align*}
    Also, notice that in the Vertex Cover instance, the minimum vertex cover must have large cardinality by the bound on node degrees: $|S^*| \ge \frac{1}{3} m$.
    Thus, an $\alpha$-approximation for the Keychain problem implies:
    \begin{align*}
        &\left(m - \frac{1}{2}\E|S_{\pi}|\right) \ge \alpha \left(m - \frac{1}{2} |S^*|\right) \\
        &2\left(1 - \alpha\right)m + \alpha |S^*| \ge \E|S_{\pi}| \\
        &\left(6\left(1 - \alpha\right) + \alpha\right) |S^*| \ge \E|S_{\pi}| \\
        &\left(6 - 5\alpha\right) |S^*| \ge \E|S_\pi|
        \tag{$3|S^*| \ge m$}
    \end{align*}

    By \cite{alimonti2000cubic}, there is a universal constant $\lambda \in (0,1)$ such that no algorithm computes vertex cover $S$ guaranteeing (assuming $P\neq NP$):
    \begin{align*}
        |S| \le (1 + \lambda) |S^*|
    \end{align*}
    Thus, if the Keychain Problem admits an approximation with a constant factor of $\alpha \ge 1 - \frac{1}{10}\lambda$, we can compute a policy $\pi$ such that $\E|S_{\pi}| \le (1 + \frac{1}{2}\lambda) |S^*|$.
    Applying Markov's inequality, we see that a randomly sampled vertex cover for policy $\pi$ will likely have low cardinality:
    \begin{align*}
        \Pr[|S_{\pi}| \ge (1 + \lambda) |S^*|] \le \frac{\E|S_{\pi}|}{(1 + \lambda) |S^*|} \le \frac{2 + \lambda}{2 + 2\lambda} < 1
    \end{align*}
    Since we know that $\lambda \in (0, 1)$ is a fixed universal constant strictly between $0$ and $1$.
    Running $O(\frac{1}{\lambda} \ln(\frac1\delta)) = O(\ln(1/\delta))$ trials draws $O(\ln(1/\delta))$ i.i.d. samples from $S_\pi$.
    One such sample is a low-cardinality vertex cover ($\le (1+\lambda)|S^*|$) with probability at least $1-\delta$ given $\pi$ is a good approximation.
    In this proof, we showed that a constant factor approximation (with a sufficiently large constant) for the Keychain Problem with many correct keys implies a constant factor approximation for Vertex Cover-3 with high probability.
    Such an algorithm does not exist unless $P=NP$.
\end{proof}

\section{Omitted Details and Proofs from Section \ref{sec:scenarios} }\label{sec:app_scenario} 
\subsection{Set Functions in Literature}\label{sec:app_scenario_background}
\subsubsection{Set Function Classes}
\begin{definition}[XOS Set functions]\label{def:app_scenarios_xos}
    A set function $f:2^{[n]} \to \R$ is XOS if it is the maximization of additive $w_1,\dots,w_h$ set functions called clauses:
    \begin{align*}
        f(S) = \max_{k\in[h]} w_k(S), \quad\text{for all }S\in 2^{[n]}
    \end{align*}
    Equivalently a $f$ is XOS if on any input $S$ it admits supporting prices $p_1,\dots,p_n$ that satisfy the following three properties:
    \begin{align*}
        p_i = 0 \text{ for all } i\notin S
        \quad\quad\quad\quad p(S) = f(S)
        \quad\quad\quad\quad p(T) \le f(T) \text{ for all } T\in 2^{[n]}
    \end{align*}
\end{definition}

\begin{definition}[Monotone Set Functions]
    A set function $f:2^{[n]}\to \R$ is monotone if it provides greater value on supersets:
    \begin{align*}
        f(S) \le f(T) \text{ for all }S\subseteq T
    \end{align*}
\end{definition}

\begin{definition}[Normalized Set Functions]
    A set function $f:2^{[n]}\to \R$ is normalized if $f(\emptyset) = 0$.    
\end{definition}

\begin{definition}[Submodular Set Functions]\label{def:app_scenario_sub}
    A set function $f:2^{[n]}\to\R$ is submodular if, on every $S\subseteq T$ and $i\in S$ it satisfies:
    \begin{align*}
        f(S) - f(S\setminus \{i\}) \ge f(T) - f(T\setminus \{i\})
    \end{align*}
\end{definition}

\begin{definition}[Supermodular Set Functions]
    A set function $f:2^{[n]}\to\R$ is supermodular if, on every $S\subseteq T$ and $i\in S$ it satisfies the inverted inequality of submodular set functions \ref{def:app_scenario_sub}:
    \begin{align*}
        f(S) - f(S\setminus \{i\}) \le f(T) - f(T\setminus \{i\})
    \end{align*}
\end{definition}

\subsubsection{Set Function Oracles}
To explicitly represent a set function $f:2^{[n]}\to\R$ requires exponential space.
Thus, we assume access to several commonly considered set function oracles.
All oracles are defined with respect to the set function $f:2^{[n]}\to \R$.

\begin{definition}[Value Oracle]
    A value oracle takes as input a set $S\in 2^{[n]}$ and outputs the value of the set function $f(S)$.
\end{definition}

\begin{definition}[Demand Oracle]
    A demand oracle takes prices $p_1,\dots,p_n \ge 0$ on elements and returns a subset $D$ of elements yielding greatest utility:
    \begin{align*}
        D\in \arg\max_{S\in^{[n]}} \left\{f(S) - p(S)\right\}   
    \end{align*}
    The sets returned by the maximizing utility are often referred to as sets in demand.
\end{definition}

\begin{definition}[Supporting Prices Oracle]
    A supporting price oracle takes as input a set $S$ and outputs supporting prices as defined in \ref{def:app_scenarios_xos}:
    \begin{align*}
        p_i = 0 \text{ for all } i\notin S
        \quad\quad\quad\quad p(S) = f(S)
        \quad\quad\quad\quad p(T) \le f(T) \text{ for all } T\in 2^{[n]}
    \end{align*}
\end{definition}

\subsection{Antichain Set Functions}\label{sec:app_scenario_antichain}
Vital to our combinatorial auction reductions are antichain functions and disjoint antichain functions, more generally.
Specifically, these set functions serve as the valuations of bidders.
In this section, we explicitly work with disjoint antichain functions.
All properties of disjoint antichain functions apply to the subsumed class of antichain functions.

Disjoint antichain functions output the maximum weight of at most $k$ disjoint antichains on a vertex-weighted forest.
We apply disjoint antichain functions to the forest representation of laminar families to find maximum weight subsets with few comparable elements.
We begin by showing that disjoint antichain functions are XOS, but not submodular nor supermodular.
To ensure that we can use disjoint antichain functions as valuations in combinatorial auctions, we also prove that all such functions are monotone and normalized.
Finally, to facilitate our reductions, we prove that this special class of set functions admits polynomial-time value, demand, and supporting price oracles.

\begin{definition}[Disjoint Antichain Function]
    A set function $f:2^{[n]} \to \R$ is a $k$-disjoint antichain function if there is a positive integer $k$, weights $w_1,\dots,w_n\in\mathbb{R}_{\ge 0}$, and type sets $T_1,\dots, T_n\subseteq T$ forming a laminar family such that
    \begin{align*}
        f(S) = \max_{A_1,\dots,A_k \in \ac(S)} w(A_1 \cup \dots \cup A_k)
    \end{align*}
    where $\ac(S) = \{ A\subseteq S \mid T_i \cap T_j = \emptyset \text{ for all } i,j\in A \text{ with } i\neq j\}$ is the collection of antichains induced by laminar family $\{T_i\}_{i\in S}$.
    $1$-disjoint antichain functions are also called antichain functions.
\end{definition}

We now prove several pertinent properties of disjoint antichain functions.
The first property is that the class of XOS functions subsumes the class of disjoint antichain functions.
To further classify disjoint antichain functions we also show that they are neither submodular nor supermodular.
The latter properties relate to the efficient implementation of value, demand, and supporting price oracles despite the exponential number of clauses in the XOS representation of disjoint antichain functions.
In each of our proofs, $f$ is a disjoint antichain function defined by a positive integer $k$, weights $w_1,\dots,w_n$, and laminar type sets $T_1,\dots, T_n$.

\begin{observation}[XOS Property]
    If a function is a disjoint antichain function, then it is also an XOS function.
\end{observation}

\begin{proof}
    By simple algebraic manipulations, we express disjoint antichain functions in their XOS representation (maximization over additive set functions):
    \begin{align*}
        f(S)
        &= \max_{A_1,\dots,A_k \in \ac(S)} w(A_1 \cup \dots \cup A_k) \\
        &= \max_{A_1,\dots,A_k \in \ac(S)} \sum_{i\in S} \left(w_i \cdot \1{i\in A_1 \cup \dots \cup A_k}\right)
    \end{align*}
    Thus, disjoint antichain functions are XOS functions with exponentially-many clauses.
    Each clause is associated with antichains a collection of antichains $A_1,\dots,A_k$ which induce weight $w_i \cdot \1{i\in A_1\cup\dots\cup A_k}$ for each element $i$.
\end{proof}

\begin{observation}[Monotonicity]
    If a function is a disjoint antichain function, then it is monotone
\end{observation}

\begin{proof}
    Consider $f(S)$ and $f(T)$ with $S \subseteq T$.
    Let $A_1^*,\dots, A_k^*$ be the antichains inducing the maximizing clause for the disjoint antichain function $f$.
    Then $A_1,\dots, A_k$ is also a feasible selection of antichains on input $B$:
    \begin{align*}
        f(S)
        = w(A_1,\dots,A_k)
        \le \max_{A_1^*,\dots,A_k^* \in \ac(T)} w(A_1^* \cup \dots \cup A_k^*)
        = f(T)
    \end{align*}
    The last inequality holds because $A_1^*,\dots,A_k^* \in \ac(T)$
\end{proof}

\begin{observation}[Normalized]
    If a function is a disjoint antichain function, then it is also normalized
\end{observation}

\begin{proof}
    We have $f(\emptyset) = w(\emptyset) = 0$ because $\ac(\emptyset) = \{\emptyset\}$
\end{proof}

\begin{observation}[Not Submodular nor Supermodular]
    There is a $1$-disjoint antichain function that is neither submodular nor supermodular.
\end{observation}

\begin{proof}
    Consider the $1$-disjoint antichain function $f:2^{[3]}\to \R_{\ge 0}$ on three elements defined by positive integer $k=1$, element weights $w_1 = 3$, $w_2 = 2$, and $w_3 = 2$, and laminar type sets $T_1 = \{a,b\}$, $T_2 = \{a\}$, and $T_3 = \{b\}$. Then $f$ is not submodular because
    \begin{align*}
        f(1,2,3) - f(1,2) = 4 - 3 \ge 3 - 3 = f(1,3) - f(1)
    \end{align*}
    Moreover, $f$ is not supermodular because
    \begin{align*}
        f(1,2) - f(1) = 3 - 3 \le 2 - 0 = f(2) - f(\emptyset)
    \end{align*}
\end{proof}

The procedure we use to answer value, demand, and supporting price queries is computing $k$ or fewer disjoint antichains with maximum weight on a forest.
We solve this problem via dynamic programming.
In our dynamic program, $\opt[i, u]$ is the weight of the $u$-disjoint maximum weight antichains in the subtree rooted at $i\in V$.
To compute $\opt[i,u]$ recursively, we consider two possible situations: $i$ is or is not one of the $u$-disjoint maximum weight antichains.

\begin{algorithm}[ht]
    \KwIn{
        positive integer $k$, vertex weights $w$, and forest $G = (V, E)$.
    }
    Initialize $\opt[i, 0] \gets 0$ for all $i\in [n]$\;
    \For{each $u=1,\dots,k$}{
        $\opt[i, u] \gets \max\left\{\sum_{j\in\child(i)} \opt[j, u], w_i + \sum_{j\in \child(i)} \opt[j, u-1] \right\}$\;
    }
    Backtrack to recover antichains $A_1,\dots,A_k$\;
    \Return Antichains $A_1,\dots,A_k$\;
    
    \caption{\textsc{Disjoint Maximum Weight Antichain}$\inner{k,w,G=(V,E)}$}
    \label{alg:app_scenarios_disjoint}
\end{algorithm}


\begin{lemma}[Disjoint Antichain Algorithm]
    Protocol \ref{alg:app_scenarios_disjoint} computes at most $k$ disjoint antichains with maximum total vertex weight in polynomial time.
\end{lemma}

\begin{proof}
    For the base case, $\opt[i, u] = 0$ because we are allowed to select $0$ antichains.
    The recursive statement considers two possibilities: including $i$ and excluding $i$ as one of the $u$ disjoint antichains.
    In the case that $i$ is included, the weight of the $u$ disjoint antichains in the subtree rooted at $i$ is the weight of $i$ plus the weight of the $u-1$ maximum weight disjoint antichains in the subtree rooted at the children of $i$.
    In the case that $i$ is excluded, we take the sum of the $u$ disjoint antichains on the subtrees rooted at the children of $i$.
    The recursive statement for \opt\, in Protocol \ref{alg:app_scenarios_disjoint} considers which of these situations yields greater weight; thus, the algorithm correctly computes every entry of the \opt\, table.
    Notably, we can recover the maximum weight antichains via a standard backtracking argument.

    Protocol \ref{alg:app_scenarios_disjoint} runs in time $O(k \cdot |V|)$ as computing the \opt\, table, and backtracking are completed in that amount of time.
\end{proof}

Value and supporting prices queries reduce to finding a maximizing clause of the XOS representation.
For disjoint antichain functions, this corresponds to identifying $k$ disjoint maximum weight antichains.
We can compute such antichains using algorithm \ref{alg:app_scenarios_disjoint}.

\begin{observation}[Value Queries]\label{obs:app_scenarios_value}
    Value queries for disjoint antichain functions can be answered in polynomial time.
\end{observation}

\begin{observation}[Supporting Price Queries]\label{obs:app_scenarios_price}
    Supporting price queries for disjoint antichain functions can be answered in polynomial time.
\end{observation}

\begin{proof}[Proof for Observations \ref{obs:app_scenarios_value} and \ref{obs:app_scenarios_price}]
    We show how to find the maximizing clause of disjoint antichain functions.
    Let $S$ be the input to the value or supporting price query.
    We take $G = (V, E)$ to be the forest with $V = S$ induced by laminar family $\{T_i\}_{i\in S}$.
    Using algorithm \ref{alg:app_scenarios_disjoint} on integer $k$, weights $w$, and forest $G = (V, E)$ recovers maximum weight antichains $A_1,\dots, A_k$ which define the maximizing clause of the disjoint antichain function $f$.

    The output of the value query is $w(A_1\cup\dots\cup A_k)$.
    In addition, the price of item $i$ for the supporting price query is $w_{i}\cdot \1{i\in A_1\cup\dots\cup A_k}$.
\end{proof}

\begin{observation}[Demand Queries]
    Demand queries for disjoint antichain functions can be answered in polynomial time.
\end{observation}

\begin{proof}
    Let $p_1,\dots,p_n \ge 0$ be the prices given to the demand query.
    We simplify demand queries:
    \begin{align*}
        &\arg\max_{S\in 2^{[n]}} \left\{f(S) - \sum_{i\in S} p_i \right\} \\
        &= \arg\max_{S\in 2^{[n]}} \left\{\left(\max_{A_1,\dots,A_k \in \ac(S)} w(A_1 \cup \dots \cup A_k)\right) - p(S) \right\} \\
        &= \arg\max_{A_1,\dots,A_k \in\ac([n])} \left\{ w(A_1\cup\dots\cup A_k) - p(A_1\cup\dots\cup A_k) \right\}
    \end{align*}
    The last equality holds because selecting nodes $i$ that are not a part of an antichain only decreases the objective.
    This node does not contribute to value and has a price $p_i \ge 0$.

    We use Algorithm \ref{alg:app_scenarios_disjoint} to find the $k$-disjoint maximum weight antichains on the forest $G = (V, E)$ induced by laminar family $T_1,\dots,T_n$ with weights $w_1 - p_1, \dots, w_n - p_n$.
    By the simplification of demand queries, an item is in demand if and only if it is contained in one of the antichains yielding maximum weight.
\end{proof}

\subsection{Proof of Proposition \ref{prop:scenarios_sampleapprox} }\label{sec_app:black-box-prop}

\begin{algorithm}[ht]
    \SetAlgoNoLine
    \KwIn{
        An instance of the Probabilistic Scenarios problem $\inner{n,m,p}$.
    }
    \KwOut{
        An admissible policy $\pi$.
    }
    Sample $h = \ceil{\frac{2 m^2 |\calO|^2 \left(\ln\left(n |\calO| / \delta\right) + 1\right) }{\epsilon^2}}$ iid. scenarios $s_1,\dots,s_h \sim p$\;
    Compute reward estimate $\hat{w}_{k,o} = \frac{1}{h} \sum_{i\in[h]} r_{k,o,s_i}$ for each $k,o$\;
    Reduce the problem to \mwlm\, as in Lemma \ref{lem:scenarios_redmwlm} using $\hat{w}$ in place of $p_o r_{k,o}$ as edge weights\;
    Reduce the MLWM instance to an instance of XOS Combinatorial Auctions (Lemma \ref{lem:scenarios_redcombo})\;
    \Return the policy $\pi$ corresponding to allocation output by $\calA$ the combinatorial auctions instance\;
    
    \caption{\textsc{Sample-Based Algorithm}($n$, $m$, $p$)}
    \label{alg:scenarios_sample}
\end{algorithm}


\begin{proof}
    We begin the proof by showing that $\hat{w}_{k,o} = p_o r_{k,o} \pm O(\epsilon / |\calO|)$ for all $k,o$ using Hoeffding's bound and union bounds.
    The proof is completed by arguing that using $\hat{w}_{k,o}$ as a proxy to compute the weight of a laminar matching introduces at most $\epsilon$ additive error.

    Define random variable $X_{k,o} = \sum_{i\in[h]} r_{k,o,s_i}$ for all $k,o$.
    This random variable satisfies $\hat{w}_{k,o} = \frac{1}{h} X_{k,o}$ and $\E[X_{k,o}] = h p_o r_{k,o}$.
    Using Hoeffding's bound implies
    \begin{align*}
        \Pr\left[|\hat{w}_{k,o} - w_{k,o}| \ge \frac{\epsilon}{2|\calO|}\right]
        = \Pr\left[|X_{k,o} - \E[X_{k,o}]| \ge \left(\frac{\epsilon}{2|\calO|}\right) h \right] \le \frac{\delta}{n|\calO|}
    \end{align*}
    Applying union bounds shows that $|\hat{w}_{k,o} - w_{k,o}| \le \epsilon / (2|\calO|)$ for all $k,o$ with a failure probability of at most $\delta$.
    Consider an arbitrary laminar matching $M$ on the constructed instance with edge weights $\hat{w}$ in the event that estimation does not fail.
    The total weight of $M$ satisfies:
    \begin{align*}
        |w(M) - \hat{w}(M)|
        = \left|\sum_{(k,o)\in M} \left(w_{k,o} - \hat{w}_{k,o}\right)\right|
        \le \frac{\epsilon |M|}{2|\calO|}
        = \epsilon / 2
        \tag{$M$ has at most $|\calO|$ edges}
    \end{align*}
    Let $M^*$ be the optimal laminar matching.
    Let $M$ be the laminar matching computed on the constructed instance of MWLM in Algorithm \ref{alg:scenarios_sample}.
    By Lemmas $\ref{lem:scenarios_redmwlm}$ and $\ref{lem:scenarios_redcombo}$ the reward of the returned policy $\pi$ satisfies
    \begin{align*}
        r(\pi)
        &= w(M) \\
        &\ge \hat{w}(M) - \frac{\epsilon}{2} \\
        &\ge \left(1 - 1/e\right)\cdot \hat{w}(M^*) - \frac{\epsilon}{2} \\
        &\ge \left(1 - 1/e\right)\cdot w(M^*) - \epsilon \\
        &= \left(1 - 1/e\right)\cdot \opt - \epsilon
    \end{align*}
    with probability at least $1-\delta$. 
\end{proof}


\subsection{Proof of Theorem \ref{thm:scenarios_hardness}}\label{sec_app_prob_hard}

\begin{definition}[Max-$(3, 2B)$-SAT]
    Given an instance of $3$-SAT with $n$ variables and $m$ clauses in which each literal appears twice among all clauses (balance) and each clause has three unique literals (uniqueness), find an assignment of variables that maximizes the number of satisfied clauses.
\end{definition}

Note that the number of variables ($n$) and clauses ($m$) appearing in a 
Max-$(3,2B)$-SAT formula must satisfy $4n = 3m$, since the left side counts the 
total number of literals (with multiplicities) when each of $2n$ literals appears 
twice in the formula, while the right side counts the total number of literals 
when each of $m$ clauses contains exactly three literals.

\begin{fact}[$(3,2B)$-SAT Hardness of Approximation due to \cite{berman2004approximation}]\label{fact:scenarios_sathard}
    In a $(3,2B)$-SAT instance with $m$ clauses, it is NP-hard to distinguish between the following cases for every $\epsilon \in (0,1/4)$:
    \begin{itemize}
        \item Complete: there is an assignment satisfying $\left(\frac{1016 - \epsilon}{1016}\right)m$ of the clauses.
        \item Sound: every assignment satisfies at most $\left(\frac{1015 + \epsilon}{1016}\right)m$ of the clauses.
    \end{itemize}
\end{fact}

In our reduction, there is a key for each literal.
The possible sequences of chains that could be realized correspond with clauses.
The main gadget of our reduction is a sequence of $n$ keychains, which is common to all scenarios.
Each of these first $n$ chains contains the positive and negative literal of a variable.
By selecting keys from these chains, a policy commits to a variable assignment.
The final keychain contains the keys corresponding to the literals of a clause selected uniformly at random.
A policy may receive a reward on the final chain only if its variable assignment specified in the first $n$ rounds satisfies the clause.

\begin{proof}
    We begin by describing how to construct an instance of the Probabilistic Scenarios problem given an instance of Max-$(3,2B)$-SAT.
    For each literal, we construct a key.
    For each clause, we construct three scenarios: this yields $3m$ total scenarios.
    Consider the clause $\ell_1 \vee \ell_2 \vee \ell_3$.
    The scenarios $s_1, s_2, s_3$ associated with this clause all have the same sequence of keychains.
    The first $n$ keychains contain the pairs of literals corresponding to a variable: $C_t = \left\{x_t, \neg x_t\right\}$ for all $t\in[n]$.
    The last chain contains the literals in the clause: $C_{n+1} = \left\{\ell_1, \ell_2, \ell_3\right\}$.
    The correct keys for the three scenarios are $k^*(s_1) = \ell_1$, $k^*(s_2) = \ell_2$, and $k^*(s_3) = \ell_3$.
    Finally, the prior $p$ is the uniform distribution over all scenarios.

    We now define a surjective map from admissible policies to $(3,2B)$-SAT assignments, in which the expected reward of the policy is determined by the number of clauses satisfied by the corresponding assignment.
    First, observe that for any admissible policy, there is an admissible policy that always plays an untested key when possible that yields at least as much expected reward.
    To further simplify our arguments, we focus on deterministic policies.
    The reduction for randomized policies follows from taking convex combinations of deterministic policies.
    Now consider a policy $\pi$ that selects keys/literals $\ell_1,\dots,\ell_n$ from the first $n$ chains.
    We map this policy to the assignment $x_t = \neg \ell_t$ for each $t\in [n]$.
    This map from policies to assignments is surjective.

    Finally, we relate the expected reward of a policy to the number of clauses satisfied
    by its corresponding variable assignment. The expected reward of policy $\pi$ may be divided into three shares: 
    \begin{enumerate}
      \item $r_1(\pi) =$ the expected number of locks opened in time steps $1,\ldots,n$. 
      \item $r_2(\pi) =$ the expected number of locks opened in step $n+1$ using a key that was
    already tested in one of the first $n$ times steps.
      \item $r_3(\pi) =$ the expected number of locks opened in step $n+1$ using a
    previously untested key.
    \end{enumerate}
    We will show that the first two of these terms are constants having no dependence on the policy,
    and that the third one depends on the number of clauses satisfied by the corresponding variable
    assignment.
    
    For any policy $\pi$, one can simulate $\pi$ 
    to precompute the set $S_{\pi}$ of keys that will be tested on the first $n$ keychains 
    assuming that none of the tested keys opens the lock. 
    Note that the contents of the first $n$ keychains are 
    the same in all scenarios, so $S_{\pi}$ is a fixed set 
    of $n$ keys that depends on $\pi$ but has no dependence
    on the scenario $s$. 
    Recall that the total number of scenarios is $3m = 4n$. By the balance property of
    Max-$(3,2B)$-SAT, each of the $2n$ keys is correct in exactly 2 of these scenarios. 
    Since $S_{\pi}$ has $n$ keys, the probability that the correct key belongs to 
    $S_{\pi}$ is $\frac{2n}{4n} = \frac12$,
    irrespective of the policy $\pi$. The number of locks opened during
    the first $n$ time steps is 1 if the correct key belongs to $S_{\pi}$ 
    and 0 otherwise, so $r_1(\pi) = \frac12$. Since the correct key always
    belongs to $C_{n+1}$ and $\pi$ is exploitative, the event that a lock
    is opened in step $n+1$ using a key that was previously tested coincides
    with the event that the correct key belongs to $S_{\pi}$, so 
    $r_2(\pi) = \frac12$. 
    
    Finally, to calculate $r_3(\pi)$, let $k_{n+1}(\pi)$ denote the key 
    that $\pi$ selects from $C_{n+1}$ assuming the first $n$ keys tested
    are not correct. Again, since the first $n$ chains are the same in all
    scenarios, the information after time $n$ is exactly the same across all scenarios. When seeing the last chain associated with a (randomly selected) clause, $\pi$ must choose the same key at time $n+1$ across all three scenarios associated with the clause ($k_{n+1}(\pi)$ is invariant across all scenarios $s$ associated with any fixed clause). 
    We can assume without loss of generality
    that $\pi$ will always select an untested key if possible. The only scenarios
    in which this is impossible are those in which the three keys on chain $C_{n+1}$
    belong to $S_{\pi}$. By our construction of the truth assignment 
    corresponding to $\pi$, this coincides with the case when this 
    truth assignment, which we denote by $\tau(\pi)$, fails to satisfy the clause 
    associated with the scenario, which we denote by $\text{cl}(s)$. 
    
    The event that $\pi$ opens the lock at time $n+1$ using a previously
    untested key is expressed by the relations $k^*(s) = k_{n+1}(\pi) \not\in S_{\pi}$. 
    We have seen that $k_{n+1}(\pi) \not\in S_{\pi}$ 
    is equivalent to the assertion that $\tau(\pi)$ satisfies 
    $\text{cl}(s)$. Furthermore, conditional on $k_{n+1}(\pi) \not\in S_{\pi}$,
    the correct key $k^*(s)$ is equally likely to be any of the three
    literals in $\text{cl}(s)$, so 
    $\Pr(k^*(s) = k_{n+1}(\pi) \, | \, k_{n+1}(\pi) \not\in S_{\pi}) = \frac13$.
    Putting these observations together, we find that 
    $ r_3(\pi) = \frac13 \cdot \frac{\text{sat}(\tau(\pi))}{m} $,
    where $\text{sat}(\tau)$ denotes the number of clauses satisfied
    by truth assignment $\tau$. 
    
    We now invoke Fact \ref{fact:scenarios_sathard} due to \cite{berman2004approximation}.
    Let $\epsilon \in (0,1/4)$ be an arbitrary constant.
    In the complete case of $(3,2B)$-SAT there is a variable assignment $\tau$ satisfying $\text{sat}(\tau) \ge \left(\frac{1016 - \epsilon}{1016}\right) m$ of the $m$ total clauses.
    Thus, there is a policy $\pi$ with expected reward at least
    \begin{align*}
        r(\pi) = r_1(\pi) + r_2(\pi) + r_3(\pi)
        = \frac12 + \frac12 + \frac{\text{sat}(\pi)}{3m} 
        \ge \frac{4064 - \epsilon}{3048}
    \end{align*}
    In the sound case, for every truth assignment $\tau$ we have $\text{sat}(\tau) \leq \left(\frac{1015 + \epsilon}{1016}\right) m$.
    Thus, the expected reward of any policy $\pi$ is bounded as follows:
    \begin{align*}
        r(\pi)
        = \frac12 + \frac12 + \frac{\text{sat}(\tau)}{3m}
        \le \frac{4063 + \epsilon}{3048}
    \end{align*}
    Thus, for any $\epsilon \in(0,1/4)$, the Probabilistic Scenarios setting is hard to approximate to a factor greater than $\frac{4063 + \epsilon}{4064 - \epsilon}$. The hardness of approximation factor in the theorem statement is obtained by applying a 
    change of variables, letting $\delta = \frac{8127 \epsilon}{4064\cdot(4064 - \epsilon)}$, which ranges from $0$ to $\frac{2031.75}{16,515,080} \approx 0.000123$ as $\epsilon$ varies from $0$ to $\frac14$.
\end{proof}

\section{Omitted Details and Proofs from Section \ref{sec:order}}\label{sec:app_order}

\subsection{NP-Hardness of Keychain Order Selection}\label{sec:app_order_hardness}
When each keychain can contain an arbitrary number of keys and the locksmith can pick the order of the keychains, we show that it is NP-Hard to determine the optimal strategy for the locksmith via a reduction from the \textit{Downwards Facing Bipartite Graph} problem. 

Consider the adjacency matrix $A_G$ for any instance graph $G$ of DFBP. It is easy to see that reordering the left nodes corresponds to permuting the rows of $A_G$, and reordering the right nodes of $G$ corresponds to permuting the columns of $A_G$. Then, finding an ordering of left and right nodes of $G$ such that there are no upward edges corresponds directly to independently permuting the rows and columns of $A_G$ such that the permuted matrix is upper triangular. 
\begin{definition} [Upper Triangular Matrix Permutation (UTMP) \cite{fertin2015obtaining,gerbner2016topological}]
    Given a binary $n \times n$ matrix, determine whether it is permute-equivalent to an upper triangular matrix.
\end{definition}

This version of DFBP is known as the \emph{Upper Triangular Matrix Permutation} (UTMP) problem, and it was shown to be NP-hard by \cite{gerbner2016topological, fertin2015obtaining}.
For ease of exposition, we will present our hardness result in terms of UTMP notation, although the proof may be converted to DFBP terminology by simply considering the underlying graph, instead of the adjacency matrix.
For the remainder of this section, we take $p$ to be the uniform distribution over keys.

\subsubsection{Notation}\label{subsubsec:app_order_hardness_notation}

To describe the UTMP problem and our reductions, it will be useful to define a few matrix-related notations. 
For an $n \times n$ matrix $M$,  its \textit{main diagonal} are the entries $M_{i, i}, \forall i \in [n]$, the \textit{upper triangle} of $M$ are the entries $M_{i, j}, \forall i \leq j$ and the \textit{lower triangle} of $M$ are the entries $M_{i, j}, \forall i \geq j$. A binary matrix $M$ has a \textit{full upper triangle} if   all the entries in the upper right triangle of $M$ are $1$ (i.e. iff $M_{i, j} = 1$, $\forall i \leq j$).
We say that a binary matrix $M$ is \textit{upper triangular} if   all its $1$-valued entries occur in the upper right triangle (i.e. iff $M_{i, j} \ne 0 \implies i \leq j$). 
Notably, a matrix with a full upper triangle is not necessarily upper triangular, nor is the reverse true. 
Finally, we say that a matrix $A$ is \emph{permute-equivalent} to $B$ if some independent combination of   row permutations and column permutations of $A$ yields $B$; in other words if there exist permutation matrices $\Pi_R$ and $\Pi_C$ such that $B = \Pi_R \, A  \, \Pi_C$.

\subsubsection{Useful Observations and Lemmas for NP-Hardness of Keychain Order Selection}\label{subsec:app_observations_for_hardness}
\begin{observation} \label{upper_bound_on_unordered_instances}
    Consider any instance of the Keychain Order Selection problem with an an equal number of keys and keychains. When $p$ is the uniform distribution, the locksmith's reward is at most $(n+1)/2$.
\end{observation}
\begin{proof}
    If the locksmith has not found the correct key by round $i$, the best she can hope for is that the key she tries shows up in every one of the future $n-i$ keychains. 
    Thus her optimal reward can be no more than:
    $$r(\pi) \leq \sum_{i \in [n]} \left( \frac{1}{n} \right) (n-i+1) = \frac{1}{n} \sum_{i \in [n]}i = \frac{n+1}{2}$$
\end{proof}

Further, observe that if it is possible to achieve the theoretical upper bound reward of $(n+1)/2$ in a uniform prior instance, then its adjacency matrix is permute-equivalent to a matrix with a full-upper triangle. 

\begin{lemma} \label{permute_equivalence_of_upper_bound}
    Consider an instance of the Keychain Order Selection problem with an equal number of keys and keychains, and let $A$ be the adjacency matrix of the instance (i.e. $A_{i, j} = 1$ iff key $i$ is on keychain $j$). The locksmith can achieve reward $(n+1)/2$ when the prior over keys is uniform if and only if $A$ is permute-equivalent to some $B$ with a full upper triangle.
\end{lemma}
\begin{proof}
    $\Leftarrow$ 
    Suppose $A$ is permute equivalent to $B$ with a full upper triangle. 
    Notice that swapping rows is equivalent to renaming the keys and swapping columns is equivalent to reordering the keychains.
    Renaming and reordering the keys and keychains of $A$ according to $B$ and playing the keys along the main diagonal of $B$ yields the locksmith the optimal reward of $(n+1)/2$.
    
    $\Rightarrow$
    The sequence $(k_1,C_1), \, (k_2,C_2), \, \ldots, (k_n,C_n)$ encodes the Locksmith's strategy; it specifies the order in which she tests the keys and the associated keychains from which they are selected.
    Observe that $k_i$ must give the locksmith exactly $n-i+1$ reward in order to achieve the theoretically optimal upper bound.
    Consider the adjacency matrix $B$ of the instance when the keys
    are listed in the order $k_1,\ldots,k_n$ and the keychains are listed in the order $C_1,\ldots,C_n$. Whenever key $i$ is tested, it must be on every subsequent keychain in order to achieve the optimal reward bound. 
    So, $B_{i, j} = 1, \forall i \leq j$ and thus $B$ has a full upper triangle.
    Further, playing the diagonal of $B$ guarantees the locksmith $(n+1)/2$ reward.
\end{proof}

\subsubsection{Reduction from Upper Triangular Matrix Permutation}\label{sec:app_order_hardness_proof}
The UTMP problem is known to be $NP$-hard \cite{fertin2015obtaining,gerbner2016topological}.
To reduce it to the \emph{Keychain Order Selection} (KOS) Problem, we first show a (potentially achievable) upper bound for any instance of KOS.
Next, we observe that any binary matrix $A$ can be interpreted as an adjacency matrix for a KOS problem, in which $A_{i, j} = 1$ if and only if key $i$ is on keychain $j$.
In this viewpoint, swapping rows of $A$ corresponds to changing the names of keys, and swapping two columns of $A$ corresponds to swapping the order of the two associated keychains.
Finally, we show that an input $M$ to UTMP can be transformed into a matrix $A$ such that KOS achieves the optimal upper bound on $A$ if and only if $M$ is upper triangular. 

\keychainorderhard*

\begin{proof}
    We will consider the corresponding UTMP instance of the DFBP problem. Given an instance $M$ of the UTMP problem, let $\overline{M}$ be the matrix resulting from flipping every entry of $M$, 
    i.e., $\overline{M} = \mathbf{1}_{n \times n} - M$.
    Let $B$ be the matrix resulting after adding a full column of $1$s to the left of $\overline{M}$ and a full row of $1$s on the bottom of $\overline{M}$.
    Finally, let $A = B^T$.
    
    We interpret $A$ as the adjacency matrix of a Keychain Order Selection instance ($A_{i, j} = 1$ iff key $i$ is on keychain $j$). 
    We let the prior for any correct key be the uniform distribution on the $n+1$ keys.
    Notice that relabeling the keys corresponds to swapping the rows of $A$ and reordering the keychains corresponds to swapping the columns of $A$.
    
    Suppose $M$ is a YES instance of the UTMP problem. 
    Then there are permutation matrices $\Pi_R, \, \Pi_C$ such that the matrix $U = \Pi_R \, M \, \Pi_C$ is upper
    triangular, i.e.~$U$ has no 1 entries below its main diagonal. Then, $V = \Pi_R \, \overline{M} \, \Pi_C$ has 
    no 0 entries below its main diagonal. Now, let $P_R, \, P_C$ be permutation matrices with $n+1$ rows and columns
    obtained from $\Pi_R, \, \Pi_C$ as follows:
    \begin{equation} \label{eq:pr_pc}
        P_R = \begin{bmatrix} \Pi_R & 0 \\ 0 & 1 \end{bmatrix}, \qquad
        P_C = \begin{bmatrix} 1 & 0 \\ 0 & \Pi_C \end{bmatrix} .
    \end{equation}
    Left multiplication by $P_R$ permutes the first $n$ rows of an $(n+1) \times (n+1)$ matrix
    according to $\Pi_R$, while leaving the bottom row fixed. 
    Right multiplication by $P_C$ permutes the last $n$ columns of an $(n+1) \times (n+1)$ matrix
    according to $\Pi_R$, while leaving the leftmost column fixed. Hence, the matrix 
    $W = P_R \, B \, P_C$ will consist of the entries of $V = \Pi_R \, \overline{M} \, \Pi_C$ 
    with a full column of 1's appended to the left and a full row of 1's appended to the
    bottom. Recalling that $V$ has no 0 entries below its main diagonal, we see that $W$ 
    has no 0 entries \emph{on or below} its main diagonal. Hence, the transpose matrix
    $W^T = P_C^T \, A \, P_R^T$ has a full upper triangle. By \Cref{permute_equivalence_of_upper_bound}, 
    the KOS instance with adjacency matrix $A$ and uniform prior over keys has a policy allowing 
    the locksmith to achieve reward $((n+1)+1)/2$. 
    

    Conversely, suppose the KOS instance with adjacency matrix $A$ and uniform prior over keys has a policy 
        achieving reward $((n+1)+1)/2$. Then, by \Cref{permute_equivalence_of_upper_bound}, it must be possible to 
        multiply $A$ on the left and right by permutation matrices to obtain a matrix with a full upper triangle. 
        Denote this factorization as $W^T = P_C^T \, A \, P_R^T$, and consider the transpose matrix
        $W = P_R \, B \, P_C$ which has no 0 entries on or below the main diagonal. In particular, the
        bottom row of $W$ consists entirely of 1's. The bottom row of $B$ also consists entirely of 1's,
        so we may assume without loss of generality that $P_R$ is a permutation matrix that fixes the bottom
        row, i.e.~$(P_R)_{n+1,n+1} = 1$. Otherwise, if $(P_R)_{i,n+1}=1$ for some $i < n+1$, then the $i^{\mathrm{th}}$ row of $P_R B$ also
        consists entirely of 1's, so we can left-multiply $P_R$ by a permutation matrix that transposes rows
        $i$ and $n+1$ to obtain a new permutation matrix $P'_R$ satisfying $(P'_R)_{n+1,n+1}=1$ and 
        $P'_R \, B = P_R \, B$. Similarly, we may assume without loss of generality that $P_C$ is a 
        permutation matrix that fixes the leftmost column, i.e. $(P_C)_{1,1}=1$. Having justified
        these assumptions about the structure of $P_R$ and $P_C$, we may represent them in the form
        given in equation~\eqref{eq:pr_pc} above, for some $n \times n$ permutation matrices 
        $\Pi_R, \, \Pi_C$. Then $\Pi_R \, \overline{M} \, \Pi_C$ consists of the first $n$ rows
        and last $n$ columns of $W$, hence it has no 0 entries below the main diagonal. Subtracting
        this matrix from the matrix $\mathbf{1}_{n \times n}$, we find that 
        $\Pi_R \, M \, \Pi_C$ has no 1 entries below the main diagonal, i.e.~it is upper triangular.
        Thus, $M$ is a YES instance of the UTMP problem, as desired.
\end{proof}

    The same reduction for the proof of Theorem \ref{thm:order_hard} implies that it is NP-Hard to determine whether a given matrix is permute equivalent to a matrix with a full upper triangle.
    Though, an algorithm for this problem would not immediately translate to an algorithm for the Keychain Order Selection problem due to the additional constraint that keys must be on the current keychain when played (see \Cref{subsec:order_opendirections} for a detailed discussion).
    \begin{corollary}
        Assuming $P \ne NP$, no algorithm can determine whether a binary square matrix $A$ is permute equivalent to a matrix with a full upper triangle. 
    \end{corollary}

\subsection{(1/2)-approximation of Keychain Order Selection}\label{sec:app_order_approx}
We begin by defining some useful notation.
We will let $\sigma$ denote an ordering over the keychains ($\sigma: \mathcal{C} \to [m]$ is a bijective map).
For a given ordering $\sigma$, we will let $-\sigma$ be the reverse ordering. 
In other words, $\sigma(C) = i$ if and only if $-\sigma(C) = m - i + 1$, $\forall C \in \mathcal{C}$.

It will be useful to think of the locksmith's policy as a tuple $\pi = (\kappa, \sigma)$, where $\sigma$ is an ordering over the keychains and $\kappa: \mathcal{C} \to [n]$ is an injective \textit{key selection} function that selects a unique key to try for each keychain, with the additional constraint that $\kappa(C) \in C, \forall C \in \mathcal{C}$.
In other words, $\kappa(C) = k$ implies that keychain $C$ plays key $k$, regardless of its position in the ordering over $\mathcal{C}$.
Finally, we let $r_{k, C,\sigma}$ be the amount of times that key $k$ shows up in keychains on or after $C$ in the order $\sigma$.
Formally, $r_{k, C,\sigma} := \sum_{C':\sigma(C') \geq \sigma(C)}\textbf{1}\{k \in C'\}$.
The locksmith's goal is to play a policy $\pi = (s, \sigma)$ maximizing her expected reward:
\begin{align*}
    r(\pi) = r(\kappa, \sigma) = \sum_{C \in \mathcal{C}} \left[p_{\kappa(C)} \sum_{C':\sigma(C') \geq \sigma(C)}\textbf{1}\{\kappa(C) \in C'\} \right] = \sum_{C \in \mathcal{C}} p_{\kappa(C)} r_{\kappa(C), C, \sigma}
\end{align*}

\orderhalfapprox*
\begin{proof}
    Our algorithm selects an arbitrary ordering over the keychains $\sigma$ and its reverse $-\sigma$.
    Then, we solve the Keychain problem with fixed order $\sigma$ and $-\sigma$, obtaining injective functions $\kappa_\sigma$ and $\kappa_{-\sigma}$. 
    We return the better of the two solutions.
    Note that it is sufficient to compare the output of our algorithm to the best deterministic policy because the expected reward of a randomized policy is always a convex combination of the expected rewards of its constituent deterministic policies, so an optimal deterministic policy achieves at least as much reward as any randomized policy.
        
    Suppose the locksmith is forced to play a fixed selection function $\kappa$ and let $N_k$ denote the number of times that key $k$ occurs over all keychains.
    Now, notice that for every ordering $\sigma$ and played\footnote{This is only true for keys that are played; unplayed keys will receive 0 reward in both orderings $\sigma$ and $-\sigma$.} key $k \in [n]$ when using selection function $\kappa$, $r_{k, C,\sigma} + r_{k, C,-\sigma} = N_k+1$.
    Aggregating over all the keys (let $\sigma^*_\kappa$ be the best ordering for $\kappa$), 
    \begin{align*}
        r(\kappa, \sigma)  + r(\kappa, -\sigma)
        &= \sum_{C \in \mathcal{C}} p_k r_{\kappa(C), C,\sigma} + \sum_{C \in \mathcal{C}} p_k r_{\kappa(C), C,-\sigma} \\ 
        &= \sum_{C \in \mathcal{C}} p_{\kappa(C)} (N_{\kappa(C)}+1) 
        \geq r(\kappa, \sigma^*_\kappa)
    \end{align*}
    The last inequality follows from the fact that the best ordering for $\kappa$ can only observe at most $N_k$ reward for every key it plays.
    Notice that the optimal policy selects the optimal selection function $\kappa^*$ as well as the optimal ordering $\sigma^*$.
    Since $\kappa^*$ is a selection function, by our previous argument:
    $$r(\kappa^*, \sigma^*) \leq r(\kappa^*, \sigma) + r(\kappa^*,-\sigma) \leq r(\kappa_\sigma, \sigma) + r(\kappa_{-\sigma}, -\sigma)$$
    where $\kappa_\sigma$ and $\kappa_{-\sigma}$ denote the optimal selection functions for $\sigma$ and $-\sigma$, respectively.
    Thus, $r(\kappa_\sigma, \sigma) \geq \frac{1}{2}r(\kappa^*, \sigma^*)$ or $r(\kappa_{-\sigma}, -\sigma) \geq \frac{1}{2}r(\kappa^*, \sigma^*)$. 
\end{proof}

\section{Weighted Online Bipartite b-Matching Algorithms}\label{subsec:app_scenarios_obm}
We apply our \mwlm\, and XOS combinatorial optimization results to the problem of Stochastic Maximum Weight Online Bipartite One-Side $b$-Matching (\wobm).
We prove all our approximation results with respect to $b$-matchings as they imply the same results in the special case of standard matchings constraints.
First, we provide an approximation-preserving reduction from \wobm\, to maximum weight laminar $b$-matching (\mwlbm).
In this generalization of Maximum Weight Laminar Matching, left node $i$ may have up to $b_i$ neighbors of each type for prespecified capacities $b$.
Lastly, we reduce \mwlbm\, to XOS Combinatorial Auctions with disjoint antichain valuations in an approximation-preserving way.
Since such combinatorial auctions admit a $(1-1/e)$-approximation, we recover a $(1-1/e)$-approximation for \wobm.

\begin{restatable}[\wobm\, Approximation]{proposition}{wobmapprox}\label{thm:app_scenarios_approx}
    There is a $\left(1 - 1/e\right)$-approximation for the \wobm\, Problem, where $n$ is the number of offline nodes.
\end{restatable}

\subsection{Setting}
There is a set of left nodes $n$, right nodes $m$, and a prior $p\in\Delta(\R_{\ge 0}^{n\times m})$ over a finite number of weighted edge sets.
Finally, there are capacities $b_i$ on each left node $i$.
In online fashion over right nodes $j$, the edges $w_{1,j},\dots,w_{n,j}$ are revealed.
Upon observing the edges incident on right node $j$, the matcher must irrevocably select a single left node $i$ to match to $j$ or choose to match nothing.
The matcher must maintain a one-sided $b$-matching in which each left node $i\in[n]$ has at most $b_i$ neighbors.

\begin{protocol}[ht]
    \KwIn{
        An instance of Weighted Online Bipartite One-Sided $b$-Matching $(n,m,p)$.
    }
    Initialize matching $M \gets \emptyset$\;
    Initialize neighbor count $c_i \gets 0$\;
    \For{each right node $j=1,\dots,m$}{
        The environment reveals weights $w_{1,j},\dots,w_{n,j}$\;
        The matcher selects an edge $(i,j)$ such that $c_i + 1 \le b_i$\;
        Update $M\gets M\cup\{ i \}$ and update $c_i \gets c_i + 1$\;
    }
    \caption{\textsc{Stochastic Weighted Online Bipartite One-Sided $b$-Matching}}
\end{protocol}


The goal of the matcher is to play a policy that competes with a Bayes' optimal solution.
Notably, we allow our algorithms to run polynomial in the support size of $p$.

\subsection{Policy Characterization}
In this section, we describe our representation of a \wobm\, policy.
We assume familiarity with Section \ref{sec:scenarios}, as it contains similar observations and objects.
\wobm\, is a Markov Decision Process with exponentially many states.
Each state is a prefix of observed weight sequences, called information sets, and the previously selected edges.
By \cite{puterman2014markov}, this implies that there is a deterministic Bayes' optimal policy for \wobm.
Such deterministic policies implicitly encode previously selected edges in this MDP.
Thus, a state of the MDP is simply an information set.
We represent deterministic policies as mappings from information sets to an offline node to select.
A policy is admissible only if it allocates at most $b_i$ information sets associated with any realizable edge set to offline node $i$.

We introduce some notation to formalize our policy characterization.
The collection of information sets is denoted as $\calO = \{w_{:,1:j}\}_{j\in [m], w\in \supp(p)}$.
For information set $o = w_{:, 1:j}$, we use $w_{i,o} = w_{i,j}$ to denote the weight of matching offline node $i$ to the online node $j$ that arrived to yield information set $o$.
We call a realization of edge weights a scenario and use $\calS = \supp(p)$ to denote the set of all scenarios.
The realized weights in scenario $s$ are denoted as $w^s$
We use $o\in P(s)$ to denote that information set $o\in\calO$ is observed in scenario $s$.
A scenario $s$ is said to be consistent with information set $o$ when $o$ is observed in $s$; we use $s\in C(o)$ to denote that $s$ is consistent with $o$.
The probability of a scenario being realized is $p_s$ and the probability that an information set $o\in\calO$ is observed at some point in the \wobm\, game is $p_o = \sum_{s\in C(o)} p_o$.
A deterministic \wobm\, policy is a mapping $\pi:\calO \to[n]$.
A policy is admissible if $|\{o\in P(s) \mid \pi(o) = i\}| \le b_i$ for all scenario $s$.
Finally the reward of a policy is:
\begin{align*}
    r(\pi)
    &= \E_{s\sim p} \left[\sum_{o\in P(s)}\sum_{i\in [n]} w_{i,o}^s \cdot\1{\pi(o) = i} \right] \\
    &= \sum_{i,o} p_o w_{i,o} \1{\pi(o) = k}
\end{align*}

\subsection{Reduction: Maximum Weight Laminar b-Matching}
We now present the Maximum Weight Laminar b-Matching Problem and provide an approximation-preserving reduction from \wobm\, to \mwlbm.
An instance of \mwlbm\, is defined on a bipartite graph $G = (L\cup R, E)$, edge weights $\hat{w}$, laminar type sets $T_1,\dots, T_n$, and left node capacities $b$.

\begin{definition}[Laminar $b$-Matching]
    Given an instance of \mwlbm, a laminar $b$-Matching is an edge set $M\subseteq E$ satisfying:
    \begin{itemize}
        \item Right nodes have at most one neighbor: $|N_M(j)|\le 1$, \quad $\forall j\in R$
        \item Left node $i$ have at most $b_i$ neighbors of each type: $|N_M(i) \cap R(t)| \le b_i$, \quad $\forall i\in L$, types $t$
    \end{itemize}
    where $R(t) = \{j\in R\mid t\in T_j\}$ is the set of right nodes with type $t$.
\end{definition}

\begin{definition}[Maximum Weight Laminar $b$-Matching]
    The goal of \mwlbm is to, given an instance, compute a laminar $b$-Matching $M$ with maximum weight: $\hat{w}(M) = \sum_{(i,j)\in M} \hat{w}_{i,j}$.
\end{definition}

With \mwlbm\, defined, we now provide an approximation preserving reduction from \wobm\, to \mwlbm.
In this reduction, offline nodes correspond to left nodes, and information sets correspond to right nodes.
The type of a right node is the set of scenarios it is consistent with.
An edge in a laminar matching in this reduction encodes a mapping from an information set to an offline node.
With carefully selected edge weights, the weight of any laminar matching will be the expected reward of the corresponding \wobm\, policy.

\begin{lemma}[Reduction to \mwlbm]\label{lem:app_scenarios_mwlbm}
    If there is a polynomial time $(\alpha,\epsilon)$-approximation for the \mwlbm, Problem, then there is a polynomial time algorithm for the \wobm\, problem. 
\end{lemma}

\begin{proof}
    Let $n,m,p,b$ be an instance of \wobm.
    We construct an instance of \mwlbm\, as follows.
    Let 
    $G = (L\cup R, E)$ be a complete bipartite graph with left nodes corresponding to offline nodes ($L = [n]$) and right nodes corresponding to information sets ($R = \calO$).
    The weight of edge $(i,o)\in E$ where left node $i$ corresponds with offline node $i$ and right node $j$ corresponds to information set $o$ is:
    \begin{align*}
        \hat{w}_{i,o} = p_o w_{i, o}  
    \end{align*}

    We map a laminar $b$-matching $M$ to the policy $\pi$ with $\pi(o) = i$ when $(i,o)\in M$.
    This mapping is a bijection.
    The proof is completed by showing that a laminar $b$-matching has equivalent weight to the expected reward of its policy counterpart $\pi$:
    \begin{align*}
        \hat{w}(M)
        = \sum_{(i,o)\in M} \hat{w}_{i,o}
        = \sum_{i,o} p_o w_{i,o} \1{\pi(o) = i} 
        = r(\pi)
    \end{align*}
\end{proof}

\subsection{Reduction: XOS Combinatorial Auctions}
In this section, we reduce \mwlbm\, to combinatorial auctions with disjoint antichain valuations.
We assume familiarity with combinatorial auctions and disjoint antichain functions.
See Section \ref{subsubsec:scenarios_combo} and Section \ref{sec:app_scenario_antichain} for background about combinatorial auctions and disjoint antichain functions, respectively.

In this reduction, each left node is a bidder and each right node is an item.
The valuation of each bidder is a disjoint antichain function, where the number of antichains is set to the capacity of the corresponding left node.
By ignoring allocated items with supporting price equal to $0$, we recover a laminar matching.
By carefully selecting the weights of the disjoint antichain function, we can ensure an allocation has social welfare equal to the weight of its corresponding laminar matching.

\begin{lemma}[Reduction to Combinatorial Auctions]\label{lem:app_scenarios_combo}
    A polynomial-time $(\alpha,\epsilon)$-approximation for the XOS Combinatorial Auctions Problem implies that the \mwlbm\ Problem admits a polynomial-time $(\alpha,\epsilon)$-approximation algorithm.
\end{lemma}

\begin{proof}
    We convert an instance of \mwlbm\, into an instance of XOS Combinatorial auctions.
    For each left node and right node, construct a bidder and an item, respectively.
    The valuation of bidder (left node) $i$ for a bundle of items (subset of right nodes) $S$ is:
    \begin{align*}
        v_i(S) = \max_{A_1,\dots,A_k\in\ac(S)} \hat{w}_i(A_1\cup\cdots\cup A_k)
    \end{align*}
    where $k = b_i$ is the capacity of left node $i$ and $\hat{w}_i(A) = \sum_{j\in A} \hat{w}_{i,j}$ are the weights of edges incident on left node $i$.

    Now we define a surjective mapping from allocations to laminar $b$-matchings.
    Fix an allocation $S_1,\dots,S_n$ on bidders.
    Let $\hat{S}_i\subseteq S_i$ be the items allocated to bidder $i$ with positive supporting price.
    Notice that the items (right nodes) with the same type appear at most $b_i$ times because $\hat{S}_i$ has non-zero supporting price on at most $b_i$ antichains.
    Then allocation $S_1,\dots,S_n$ corresponds to the laminar $b$-matching in which left node (bidder) $i$ has neighbors $\hat{S}_i$.
    To prove this map is surjective, consider an arbitrary laminar $b$-matching $M$ in which left node $i$ has neighbors $N_i$.
    The allocation $N_1,\dots,N_n$ maps to laminar $b$-matching $M$.

    The proof is completed by showing that the social welfare of an allocation is equal to the expected reward of the laminar $b$-matching associated with the allocation.
    Consider allocation $S_1,\dots,S_n$ in which bidder (left node) $i$ has items $\hat{S}_i$ with non-zero supporting price.
    Then the associated laminar $b$-matching $M$ satisfies:
    \begin{align*}
        \sum_{i} v_i(S_i)
        = \sum_{i} \hat{w}_i(\hat{S}_i)
        = \sum_{i} \sum_{j\in A} \hat{w}_{i,j} 
        = \sum_{(i,j)\in M} \hat{w}_{i,j}
    \end{align*}
\end{proof}

\subsection{A \texorpdfstring{$(1-1/e)$}--Approximation for \wobm}
We use Lemmas \ref{lem:app_scenarios_mwlbm} and \ref{lem:app_scenarios_combo} to complete the proof of Theorem \ref{thm:app_scenarios_approx}.

\wobmapprox*

\begin{proof}
    By Lemma \ref{lem:app_scenarios_mwlbm} and \ref{lem:app_scenarios_combo}, an $\alpha$-approximation for the Combinatorial Auctions Problem with XOS valuations implies an $\alpha$-approximation for \wobm.
    Since Theorem 3.2 of \cite{dobzinski2010auctions} provides a $(1-1/e)$-approximation for the XOS Combinatorial Auction Problem, \wobm\, admits a $(1-1/e)$-approximation.
\end{proof}

\subsection{Hardness of Approximation}\label{sec_app_philosopher_hard}
We know that \wobm\, is APX-hard using the special case of the standard matching constraint ($b_i = 1$ for all $i$).
Since this proof closely resembles that of \Cref{thm:scenarios_hardness}, we state our claim and outline the reduction.

\begin{proposition}[\wobm\, Hardness of Approximation]\label{prop:app_scenarios_wobmhard}
    Assuming $P \neq NP$, there is no polynomial time algorithm for the \wobm\, Problem that guarantees a $\left(\frac{4063}{4064} + \epsilon\right)$-approximation for any $\epsilon \in(0,1/4)$.
\end{proposition}

\textbf{Proof Outline: }We reduce from the 3-SAT instance in \cite{berman2004approximation} that satisfies uniqueness and balance.
Given a 3-SAT instance with $n$ variables (and $4m/3$ clauses), we construct a graph with $2n$ offline and $n+1$ online nodes.
Each offline node corresponds to a literal.
The first $n$ neighbor sets deterministically include edges with weight $\frac{1}{n}$ to the pairs of literals corresponding to a variable: online node $i\le n$ has edges to the offline nodes representing $x_i$ and $\neg x_i$.
The last online node is determined by a clause drawn uniformly at random.
The node contains edges to the offline nodes associated with literals in the clause with weight $\frac{1}{3}$.
The proof is completed by making similar arguments to \Cref{thm:scenarios_hardness}.

\section{AI Software Disclosure}\label{sec:app_ai_disclosure}
We used Refine.ink \cite{refine_ec} to identify notational inconsistencies, minor mathematical bugs, and clarity issues, which the authors corrected.
The authors have verified the AI-generated content to the best of their ability.

\end{document}